\documentclass[a4paper,UKenglish,cleveref, autoref, thm-restate]{lipics-v2021}

\hideLIPIcs  

\usepackage[T1]{fontenc}
\usepackage[utf8]{inputenc}
\usepackage{microtype}
\usepackage{amsmath,amssymb,amsthm,mathtools}
\usepackage{graphicx}
\usepackage{xspace}
\usepackage{tikz}
\usepackage{xcolor}
\usepackage{subcaption}
\usepackage{thmtools}
\usepackage{thm-restate}
\usepackage{todonotes}

\newcommand\indeg{{\textnormal{in}}}
\newcommand\out{{\textnormal{out}}}

\newcommand\edge[2]{\ensuremath{\{#1,#2\}}}

\usetikzlibrary{calc,decorations.pathmorphing,shapes.misc, fit}
\definecolor{darkred}{rgb}{0.62,0.13,0.28}
\definecolor{darkgreen}{rgb}{0.01,0.53,0.47}

\newcommand{\N}{\mathbb{N}}
\newcommand{\dist}{\operatorname{dist}}
\newcommand{\FO}{\ensuremath{\mathsf{FO}}\xspace}
\newcommand{\MSO}{\ensuremath{\mathsf{MSO}}\xspace}

\newcommand{\cw}{\operatorname{cw}}

\newcommand{\W}{\textsf{W}}
\newcommand{\NP}{\textsf{NP}}
\newcommand{\XP}{\textsf{XP}}

\newcommand{\PVC}{\textsc{Partial Vertex Cover}\xspace}
\newcommand{\PVCD}{\textsc{Partial Vertex Cover Discovery}\xspace}
\newcommand{\VCD}{\textsc{Vertex Cover Discovery}\xspace}

\newcommand{\FOCV}{\textsc{FO Cost-Value Decision}\xspace}

\newcommand{\FOVD}{\textsc{FO Value Discovery}\xspace}

\newcommand{\GK}{\textsc{Group Knapsack}}

\newcommand{\Ufam}{\mathcal{U}}

\newcommand{\Cc}{\mathscr{C}}

\providecommand{\Z}{\mathbb Z}

\newcommand{\val}{\operatorname{val}}
\newcommand{\cost}{\operatorname{cost}}

\providecommand{\cov}{\operatorname{cov}}

\providecommand{\cost}{\operatorname{cost}}

\providecommand{\Z}{\mathbb Z}

\renewcommand{\phi}{\varphi}

\title{FO Value Discovery and Partial Vertex Cover Discovery} 

\titlerunning{FO Value and Partial Vertex Cover Discovery} 

\author{Enna Gerhard}{University of Bremen, Germany}{}{gerhard@uni-bremen.de}{}
\author{Stephanie Maaz}{University of Waterloo, Canada}{}{smaaz@uwaterloo.ca}{}
\author{Pascale Schott}{University of Bremen, Germany}{}{schott@uni-bremen.de}{}
\author{Sebastian Siebertz}{University of Bremen, Germany}{}{sebastian.siebertz@uni-bremen.de}{}
\author{Jan Wodtke}{University of Bremen, Germany}{}{jan.wodtke@uni-bremen.de}{}
\authorrunning{E. Gerhard, S. Maaz, P. Schott, S. Siebertz, and J. Wodtke}

\Copyright{Enna Gerhard, Stephanie Maaz, Pascale Schott, Sebastian Siebertz, and Jan Wodtke}
\category{}
\relatedversion{}
\supplement{}
\funding{}
\acknowledgements{}


\category{} 

\relatedversion{} 

\nolinenumbers 

\EventEditors{John Q. Open and Joan R. Access}
\EventNoEds{2}
\EventLongTitle{42nd Conference on Very Important Topics (CVIT 2016)}
\EventShortTitle{CVIT 2016}
\EventAcronym{CVIT}
\EventYear{2016}
\EventDate{December 24--27, 2016}
\EventLocation{Little Whinging, United Kingdom}
\EventLogo{}
\SeriesVolume{42}
\ArticleNo{23}

\begin{document}

\maketitle

\begin{abstract}
We study solution discovery in the token-sliding model from a logical and cost-value optimization perspective. 
In solution discovery, we are given a graph, an initial placement of $k$ tokens, and a movement budget $b$. 
The task is to find a reachable target configuration satisfying a prescribed condition.

Our results are inspired by \textsc{Partial Vertex Cover Discovery}, where the condition is that the~$k$ tokens cover at least $t$ edges of the input graph.
This objective is not merely a sum of independent occupied vertex contributions: each selected vertex contributes its degree, but edges with both endpoints selected have to be subtracted once. 
To capture this phenomenon, we introduce \textsc{FO Value Discovery}, an optimization problem in which the value of a selected tuple is given by unary vertex weights together with first-order definable correction terms.
This framework generalizes \textsc{Boolean FO Discovery} and contains \textsc{Partial Vertex Cover Discovery} as a quantifier-free special case.

We further generalize the setting to \textsc{FO Cost-Value Decision}, where vertices carry both costs and values, and the task is to decide whether there is a tuple whose first-order value expression reaches a prescribed value threshold while respecting a cost bound.
We prove a conditional meta-theorem based on two algorithmic ingredients: local FO cost-value decision and an anchored weighted multicolored distance-independence problem.
We then prove that these assumptions are satisfied in several restricted graph classes and fragments of first-order logic. 

Finally, we study the parameterized complexity of \textsc{Partial Vertex Cover Discovery} and \textsc{Vertex Cover Discovery}. As a consequence of the logical meta-theorems, we obtain fixed-parameter tractability of \textsc{Partial Vertex Cover Discovery} on several graph classes, including classes of locally bounded cliquewidth.
We also show that \textsc{Partial Vertex Cover Discovery} is \W[1]-hard parameterized by $k+b$ and fixed-parameter tractable on $d$-degenerate graphs parameterized by $k+d$. For \textsc{Vertex Cover Discovery}, we prove \NP-hardness on planar graphs, \W[1]-hardness parameterized by the clique cover number, even when a clique cover is supplied with the input, and \W[1]-hardness with respect to parameter cutwidth.
\end{abstract}

\section{Introduction}
\label{sec:introduction}

Many graph problems ask for a set (or tuple) of vertices with a prescribed property, such as being a vertex cover or an independent set.
In \emph{solution discovery}, one does not construct such a solution from scratch. 
Instead, the input already contains an initial placement of tokens, and the task is to decide whether the tokens can be moved, within a given budget, to a placement that satisfies the desired property. 
This captures situations in which a current undesired assignment or configuration is established, and only a limited number of local repair operations are allowed to reach a feasible solution.

The solution discovery viewpoint was formalized by Fellows et al.~\cite{fellows2025onsd} and is closely related to earlier work on minimizing movement to reach a desired configuration, including the work of Demaine et al.~\cite{DBLP:conf/soda/DemaineHMSOZ07}.~It has since been studied for polynomial-time solvable graph problems~\cite{GroblerMMMRSS24}, from the viewpoint of kernelization~\cite{GroblerMMNRS24}, through algorithmic meta-theorems~\cite{bousquet2025algorithmic}, for several classical vertex-subset problems~\cite{saito2026solution}, and in a two-graph model that separates feasibility from movement~\cite{von2026separating}. 
Solution discovery is also related to combinatorial reconfiguration~\cite{Nishimura18}. 
Unlike reconfiguration, solution discovery starts from an infeasible configuration and imposes feasibility only on the final configuration, not on the intermediate configurations. 

Existing formulations of solution discovery are mainly Boolean: a target placement either satisfies the prescribed property or it does not. Many natural discovery tasks, however, are quantitative. One may want to optimize the quality of the reachable target, and this quality need not decompose into independent vertex contributions. Instead, equality patterns, adjacency patterns, or more general first-order (\FO) definable patterns of the selected tuple may contribute bonuses, penalties, or render the tuple of chosen vertices forbidden. 

Our guiding example is \PVC. Given a graph and integers $k$ and~$t$, the question is whether there is a set of at most $k$ vertices covering at least $t$ edges. In the discovery version studied here, the number of tokens fixes the size of the final set. Every selected vertex contributes its degree, but an edge is counted twice if both its endpoints are selected. Hence, for a set $X\subseteq V(G)$, the number of edges covered by $X$ is
\begin{equation*}
\cov_G(X)=|{e\in E(G):e\cap X\neq\emptyset}|=\sum{v\in X}\deg_G(v)-|E(G[X])|.
\end{equation*}

Thus, the value of a chosen set is a sum of unary vertex weights, together with a correction term determined by the adjacency pattern of the tuple.
More generally, we also allow unary costs for the tuple coordinates. This leads to \FOCV.
A formula $\varphi(\bar x)$ with free variables $\bar x$ defines, in a graph $G$, the set $\{\bar v\in V(G)^{|\bar x|}:G\models\varphi(\bar v)\}$.
We write $\Z_{-\infty}$ for $\Z\cup\{-\infty\}$ and $\Z_{+\infty}$ for $\Z\cup\{+\infty\}$.

\begin{definition}[\FOCV]
\label{def:weighted-fo-cost-value-optimization}
An instance of \FOCV consists of a graph $G$, cost functions $c_1,\ldots,c_k\colon V(G)\to\Z_{+\infty}$, value functions $w_1,\ldots,w_k\colon V(G)\to\Z_{-\infty}$, first-order formulas $\varphi_1(\bar x),\ldots,\varphi_m(\bar x)$, where $\bar x=(x_1,\ldots,x_k)$, and corresponding constants $\alpha_1,\ldots,\alpha_m\in\Z_{-\infty}$. The number $k$ is called the \emph{arity} of the instance.

We assume that the formulas $\varphi_1,\ldots,\varphi_m$ define a partition of $V(G)^k$ in every graph~$G$. 
We are also given a cost bound $C\in\Z$ and a value
threshold $W\in\Z$.
For a tuple $\bar v=(v_1,\ldots,v_k)\in V(G)^k$, define
\begin{equation*}
\operatorname{cost}(\bar v)=\sum_{i=1}^k c_i(v_i), \qquad
\operatorname{val}(\bar v)=\sum_{i=1}^k w_i(v_i)+\alpha_j,
\end{equation*}
where $j$ is the unique index such that $G\models\varphi_j(\bar v)$. 

In this decision variant, the task is to decide whether there exists a tuple $\bar v\in V(G)^k$ with $\val(\bar v)\geq W$ and $\cost(\bar v)\leq C$.
In the optimization variant, the task is to compute a tuple of optimum value subject to the cost bound.
\end{definition}

Note that the partition assumption is notational only. 
Any finite family of formulas can be refined into the Boolean partition generated by all truth assignments to these formulas.
We call the family $(\phi_1,\alpha_1), \ldots, (\phi_m,\alpha_m)$ appearing in an instance of \FOCV an \emph{\FO value expression}.

\medskip
This framework contains Boolean first-order constraints in unweighted graphs. 
If $\psi(\bar x)$ is an FO formula, use the two cases $\psi(\bar x)$ and $\neg\psi(\bar x)$, assign correction $0$ to the first case and~$-\infty$ to the second, and set all vertex weights to zero. 
Tuples of value at least $0$ are then exactly the tuples satisfying $\psi$.
\medskip

It also subsumes the solution discovery variant, which we call \FOVD. 

\begin{definition}[\FOVD]
\label{def:FO-value-discovery}

Let $\Theta=(\Theta_k)_{k\ge 1}$ be a computable family of \FO value expressions.
An instance of \textsc{FO Value Discovery} consists of
\begin{enumerate}
    \item a graph $G$,
    \item an initial placement of $k$ tokens given by a tuple
    $\bar s=(s_1,\ldots,s_k)\in V(G)^k$,
    \item vertex value functions
    $w_1,\ldots,w_k\colon V(G)\to\Z_{-\infty}$,
    \item for every $i\in\{1,\ldots,k\}$ the movement cost function is defined as 
$c_i(v):=\dist_G(s_i,v)$,
    \item a movement budget $b\in\N$, and
    \item a value threshold $W\in\Z$.
\end{enumerate}

Using these cost functions together with the \FO value expression
$\Theta_k$, we obtain an instance of \FOCV.
The question hence  is whether there exists a tuple
$\bar t=(t_1,\ldots,t_k)\in V(G)^k$
such that $\cost(\bar t)\le b$
and $\val(\bar t)\ge W$.
Equivalently, we ask whether the initial token placement can be moved,
within total cost at most $b$, to a tuple whose value is at least $W$.
\end{definition}

The special case in which $\Theta_k$ consists only of the formula
$\phi_1(x_1,\ldots,x_k)\equiv \top$ and the correction value
$\alpha_1=0$ is called \textsc{Value Discovery}.
In this case every discovered tuple is feasible, and the task is only to
move the initial tokens within the given budget so as to obtain a tuple
whose value is at least the given threshold.
\medskip

As a special case we also obtain our motivating example \PVCD. 
For solution size $k$, set $w_i(v):=\deg_G(v)$ for all $i\in [k]$. 
The correction terms defined by $\phi_1,\ldots, \phi_m$ are the equality and adjacency patterns of $x_1,\ldots,x_k$, which are definable by atomic first-order formulas. 
Note that $m\leq f(k)$ for a function $f$. 
For an equality and adjacency pattern $\phi_j$ we assign the correction value $\alpha_j=-|\{\{i,j\}:1\le i<j\le k$, $E(x_i,x_j)\}|$.
To obtain the discovery version, we are furthermore given an initial tuple $\bar s=(s_1,\ldots, s_k)$ and the task is to move the tokens to a desired solution $\bar t=(t_1,\ldots, t_k)$. 
Moving a token along one edge costs one unit. 
We write $\dist_G(u,v)$ for graph distance, with value $\infty$ if $u$ and $v$ lie in different connected components. 
For the given tuple $\bar s=(s_1,\ldots, s_k)$ we assign the cost functions $c_i(v)=\dist_G(s_i,v)$ as the number of steps it takes to move token~$s_i$ to the vertex $v$. 
For the discovered tuple $\bar t$, 
\begin{equation*}
        \cost(\bar t)=\dist_G(\bar s,\bar t):=\sum_{i=1}^k\dist_G(s_i,t_i).
\end{equation*}

\medskip

\paragraph*{Contribution 1: logical meta-theorems for \FOCV and\\ \FOVD.}

We isolate two algorithmic ingredients that are sufficient for fixed-parameter tractability of \FOCV.
The first is a local version of \FOCV: for every fixed radius, we assume that \FOCV can be solved when all tuple entries are required to lie in a bounded-radius neighborhood of a given center. We allow the value expression to mention this center as an additional parameter, which is useful for the locality-based reduction. 

\begin{definition}[Efficient \textsc{Local} \FOCV]
\label{def:efficient-local-FOCV}
A graph class $\Cc$ admits \emph{efficient \textsc{Local} \FOCV} if, for every
$r\in\N$, the restriction of \FOCV to instances with the following additional
structure is fixed-parameter tractable, parameterized by $r$, the arity, and
the defining value expression:
\begin{itemize}
\item the input contains a distinguished center $a\in V(G)$;
\item every feasible tuple $\bar v$ is required to lie in $N_r^G[a]^k$;
\item the formulas in the value expression may use the center $a$ as an
additional parameter.
\end{itemize}
Equivalently, the value expression is of the form
$\{(\varphi_1(\bar x,z),\alpha_1),\ldots,
(\varphi_m(\bar x,z),\alpha_m)\}$,
where~$z$ is interpreted by the center $a$.
\end{definition}

The second ingredient is the global packing problem produced by the locality argument. After local candidates have been computed with their corresponding costs and values, one must choose one candidate for each local block so that the chosen vertices are pairwise far apart while the total cost and total value satisfy the prescribed bounds. We formalize this as an anchored weighted multicolored distance independence problem.

\begin{definition}[Efficient anchored weighted multicolored distance independence]
\label{def:efficient-anchored-wmdi}
A graph class~$\Cc$ admits \emph{efficient anchored weighted multicolored
distance-$r$ independence} if, for every radius~$r$, the following problem is
fixed-parameter tractable parameterized by $p+r$, where $p$ is the number of
colors.

The input consists of a graph $G\in\Cc$, pairwise disjoint candidate sets
$A_1,\ldots,A_p$, an anchor map
$a\colon A_1\cup\cdots\cup A_p\to V(G)$, cost and profit functions
$\kappa\colon A_1\cup\cdots\cup A_p\to\N,$
\mbox{$\nu\colon A_1\cup\cdots\cup A_p\to\Z$}, a profit threshold $T$, and a cost budget $B$ that is polynomially bounded in $|V(G)|$. The question is whether there are candidates $y_i\in A_i$ for all $i\in[p]$ such that $\dist_G(a(y_i),a(y_j))>r$ for all $i\ne j$, $\sum_{i=1}^p \kappa(y_i)\le B$, and
$\sum_{i=1}^p \nu(y_i)\ge T$.
\end{definition}

\begin{remark}
The ordinary vertex formulation is the special case where
$A_i\subseteq V(G)$ and $a(v)=v$.
The anchored formulation is only a mild extension of the ordinary vertex
formulation. Indeed, whenever the ordinary version is handled through the
usual FO model-checking machinery, the anchored version can be treated in
the same way: one adds suitable unary predicates to the vertices of~$G$ recording which candidates may use a given vertex as their anchor,
together with the relevant local cost and profit information. 
The \FO part
then selects the anchors and expresses the pairwise distance condition
$\dist_G(x_i,x_j)>r$, while the bounded number of selected colours and the
polynomially bounded budget allow the remaining cost-profit bookkeeping to
be handled by the standard dynamic-programming layer.

Thus, for the purposes of the graph classes considered here, the anchored
problem reduces straightforwardly to the ordinary vertex-based setting once
one allows these additional colors. We nevertheless state the assumption in the anchored form because it is the natural output of the locality reduction:
the same anchor vertex may give rise to several candidates with different
local costs and values. 
We note that the anchored weighted distance-independence problem is weaker than full \FO model checking. In principle, it may therefore be tractable on graph classes where full \FO model checking is not known, or is not expected, to be tractable. In the applications below, however, we do not derive any additional tractability
results from this distinction. 
We only use the fact that, for the graph classes under consideration, the relevant FO model-checking machinery is
available.
\end{remark}

\begin{restatable}{theorem}{conditionalfovaluediscovery}
\label{thm:conditional-FO-value-discovery}
Let $\Cc$ be a graph class that admits efficient \textsc{Local} \FOCV and efficient
anchored weighted multicolored distance independence.
Then, \FOCV is fixed-parameter tractable on $\Cc$, parameterized by the
arity $k$ and the defining value expression.
\end{restatable}

This theorem is perhaps unsurprising to those familiar with \FO model checking. 
It is essentially the standard model-checking machinery based on the locality of first-order logic enriched with the optimization component. 
We see our contribution for \FOCV as mainly conceptual. 

We instantiate the theorem on several sparse and dense graph classes. For classes of locally structurally bounded expansion, we obtain the full \FOCV result. Here, every bounded-radius neighborhood belongs to a class of structurally bounded expansion. 
Classes of bounded expansion were introduced by Ne\v{s}et\v{r}il and Ossona de Mendez and form a robust notion of uniform sparsity~\cite{nevsetvril2008grad}.
They include classes excluding a fixed minor, and classes excluding a fixed topological minor.
Structurally bounded expansion, introduced by Gajarsk\'y et al.~\cite{gajarsky2020first}, is a dense analogue that includes natural dense
classes such as map graphs and fixed powers of bounded-expansion classes.

\begin{restatable}{theorem}{sbefocv}
\label{thm:sbe-FOCV}
Let $\Cc$ be a graph class with locally structurally bounded expansion.
Then \FOCV is fixed-parameter tractable on $\Cc$, parameterized by the
arity~$k$ and the defining value expression.
\end{restatable}

The same result holds on graph classes of bounded cliquewidth.

\begin{restatable}{theorem}{cwfocv}
\label{thm:cw-FOCV}
\FOCV is fixed-parameter tractable parameterized by the arity $k$, the defining value expression, and the cliquewidth of the input
graph.
\end{restatable}

We also obtain the corresponding local version. 
A graph class has locally bounded cliquewidth if, for every radius $r \in \N$, there is a bound on the cliquewidth of all radius-$r$ closed neighborhoods in graphs from the class.

\begin{restatable}{theorem}{localcwfocv}
\label{thm:local-cw-FOCV}
Let $\Cc$ be a graph class of locally bounded cliquewidth.
\FOCV is fixed-parameter tractable on $\Cc$, parameterized by the
arity $k$ and the defining value expression.
\end{restatable}

We also prove restricted results for \FOCV on monadically stable graph classes, which constitute the current boundary of efficient \FO model checking~\cite{DreierEMMPT23,DMS23}.
Our cost-value result applies to the existential fragment of first-order logic.

\begin{restatable}{theorem}{existentialfovaluemonadicstable}
\label{cor:existential-FO-value-monadically-stable}
Let $\Cc$ be a monadically stable graph class.
The existential fragment of \FOCV is fixed-parameter tractable on $\Cc$, parameterized by the arity $k$ and the defining value expression.
\end{restatable}

For full first-order logic on monadically stable classes, we currently obtain only the following Boolean discovery result. 
Here, \textsc{Boolean FO Discovery} denotes the special case of \FOVD in which all value functions are identically zero and the value expression only enforces an \FO-definable feasibility condition. The costs are the movement distances from the initial token placement.
Even though this is only a restricted result, it is the most interesting of the logic related results. 
The key idea is that the numerical information needed locally can be encoded by boundedly many unary colors. 
After the movement costs have been capped, we fix a center and a bounded offset pattern, add unary predicates for the corresponding distance layers, and then test the resulting local formulas by first-order model checking on a monadic expansion of the input graph. 
Thus the weighted information is absorbed into the formulas, while the global choice of local candidates is handled by anchored distance independence.

\begin{restatable}{theorem}{booleanfodiscoverymonadicstable}
\label{thm:boolean-fo-discovery-monadically-stable}
Let $\Cc$ be a monadically stable graph class.
\textsc{Boolean FO Discovery} is fixed-parameter tractable on $\Cc$, parameterized by the number $k$ of tokens and the defining value expression.
\end{restatable}

We expect that \FOCV is fixed-parameter tractable on monadically stable classes in full generality. Extending the known model-checking machinery to this setting appears to require an additional weighted optimization component, and we leave this problem for future work.

\paragraph*{Contribution 2: Partial Vertex Cover Discovery.}

We study the special cases of \PVCD and \VCD in more detail.
First, we show that \PVCD is $\W[1]$-hard when parameterized by $k+b$, where $k$ is the number of tokens and $b$ is the movement budget.
For the parameter $k+\cw(G)$, where $\cw(G)$ denotes the cliquewidth of the input graph, we obtain fixed-parameter tractability as an application of our general \FOCV result.
Our logical results imply fixed-parameter tractability of \PVCD for several graph classes, including locally bounded cliquewidth / treewidth.
For $d$-degenerate graphs, we give a dedicated algorithm whose main subroutine is the \textsc{Value Discovery} problem.
Finally, for \VCD, we show that it is \NP-hard on planar graphs, $\W[1]$-hard parameterized by the clique cover number, even when a clique cover is given as part of the input, and it remains $\W[1]$-hard with respect to cutwidth.

\paragraph*{Related logical optimization formalisms.}
Several formalisms extend first-order or monadic second-order logic with counting, weights, or aggregation, including LinEMSOL optimization~\cite{CourcelleMR00,CourcelleE12}, first-order logic with counting and cardinality conditions~\cite{KuskeSchweikardtLICS2017,GroheSchweikardtPODS2018}, weight-aggregation logics~\cite{vanBergeremSchweikardtCSL2021,vanBergeremSchweikardtCSL2025}, and aggregate-query languages~\cite{TorunczykPODS2020}. 
These formalisms allow aggregate values to be defined inside formulas or terms, for example, by counting witnesses or summing weights over definable sets.

Our setting is different: the formulas in an \FO value expression are ordinary first-order formulas, and the weights are used externally to evaluate a fixed tuple. This distinction is important because existing aggregation results either apply to more restrictive graph-classes or solve different algorithmic tasks, and they do not directly yield the weighted discovery results needed here.

\paragraph*{Organization.}
We first present the results for \PVCD and \VCD, including the dedicated algorithms and hardness reductions. We then develop the general \FOCV framework and prove the logical meta-theorems that underlie the class-based tractability results highlighted earlier. This order keeps the combinatorial applications in the foreground at the cost of a few forward references to the logical machinery. 
We additionally use section specific preliminaries.

\section{Preliminaries}
\label{sec:preliminaries}
We write $\N$ for the set of non-negative integers and $\Z$ for the set of integers.
For $n\in\N$, we let $[n]\coloneqq\{1,\ldots,n\}$.
All graphs are finite, simple, and undirected.
For a graph $G$, we write~$V(G)$ and $E(G)$ for its vertex set and edge set.
For vertices $u,v\in V(G)$, we write $\dist_G(u,v)$ for their distance in $G$, with $\dist_G(u,v)=\infty$ if $u$ and $v$ lie in distinct connected components.
For a vertex $v\in V(G)$ and an integer $r\in\N$, we write $N_r^G[v]\coloneqq\{u\in V(G):\dist_G(u,v)\le r\}$.
For a set $X\subseteq V(G)$, we write $G[X]$ for the subgraph of $G$ induced by $X$.

We use standard terminology from parameterized complexity and refer to~\cite{CyganFKLMPPS15} for background. 
Unless explicitly stated otherwise, all numerical weights in this paper are encoded in unary; hence, pseudo-polynomial dependence on these values is polynomial in the input size.

\section{Partial Vertex Cover and Vertex Cover Discovery}
\label{sec:pvc-vc}
In this section, we study \VCD and \PVCD.
Notice that whenever a positive result for \PVCD is proven, the corresponding result for \VCD follows immediately.
When we only know an argument for \VCD, we state it explicitly.
For lower bounds, analogously, hardness for \VCD already implies hardness for \PVCD by setting the coverage threshold to $|E(G)|$. 
In this section, we will use the following structural graph parameters. A graph is $d$-\emph{degenerate} if every subgraph has a vertex of degree at most~$d$.
Equivalently, its vertices admit an ordering in which each vertex has at most~$d$ later neighbors. The \emph{clique cover number} of $G$ is the minimum number of disjoint cliques whose union is~$V(G)$. The \emph{cutwidth} of $G$ is the minimum, over all linear orderings $v_1,\ldots,v_n$ of $V(G)$, of the maximum number of edges with one endpoint in ${v_1,\ldots,v_i}$ and the other in ${v_{i+1},\ldots,v_n}$.

Fellows et al.~\cite{fellows2025onsd} showed that \VCD is fixed-parameter tractable when parameterized by $k$, is $\W[1]$-hard when parameterized by $b$, and admits \XP\ algorithms parameterized by treewidth. 
Subsequent work showed that it admits a polynomial kernel of size $O(k^2)$, while also proving lower bounds for parameters such as pathwidth and feedback vertex set number~\cite{GroblerMMNRS24}. 
More recently, Saito et al.~\cite{saito2026solution} proved that the problem is in~\XP\ parameterized by cliquewidth, is \NP-complete on chordal graphs and on graphs of diameter two, and is polynomial-time solvable on split graphs. 
Algorithmic meta-theorems for solution discovery provide further context: it was shown that token-sliding solution discovery for MSO$_1$-definable properties, including vertex cover, is fixed-parameter tractable on graphs of bounded neighborhood diversity~\cite{bousquet2025algorithmic}.

Finally, \Cref{cor:existential-FO-value-monadically-stable} and  \Cref{thm:cw-FOCV} apply directly to \PVCD, since the value expression for \PVCD is quantifier-free.
Hence, \PVCD is fixed-parameter tractable parameterized by~$k$ on monadically stable graph classes and on graph classes of locally bounded cliquewidth.
By \Cref{thm:cw-FOCV}, the problem is fixed-parameter tractable when parameterized by $k+\cw(G)$. 

\subsection{In FPT on $d$-degenerate graphs with parameter k}
\label{subsec:pvc-degenerate}
We first prove the positive result for degenerate graphs. The proof extends a covering argument of Panolan et al.~\cite{PanolanSST22} for \textsc{Partial Vertex Cover} on $d$-degenerate graphs.

\begin{theorem}
\label{thm:pvcd-deg-new}
\PVCD is fixed-parameter tractable when parameterized by $k+d$.
\end{theorem}

\begin{proof}
Let $(G,S,b,t)$ be an instance of \textsc{Partial Vertex Cover Discovery}, where $|S|=k$, and let $G$ be $d$-degenerate.
Compute, in polynomial time, an ordering $\lambda \colon [n] \to V(G)$ such that, for every $1 \le i \le n$, the vertex $\lambda(i)$ has at most $d$ neighbors among $\lambda(i+1),\dots,\lambda(n)$~\cite{DBLP:journals/jacm/MatulaB83}.
For a vertex $v\in V(G)$, let $N(v)$ denote its open neighborhood, and let $LN_\lambda(v)$ denote the set of at most $d$ later neighbors of $v$ in $\lambda$.
Set $\ell \coloneqq \min\left\{n,\,k+kd\right\}$.

By the construction of universal sets of Naor, Schulman, and Srinivasan~\cite{NSS95}, there exists a family $\mathcal U\subseteq 2^{V(G)}$ such that for every set $A\subseteq V(G)$ of size at most $\ell$, the family $\{A\cap U : U\in\mathcal U\}$ contains all subsets of $A$. Moreover, such a family can be constructed with
$|\mathcal U| = 2^\ell \ell^{O(\log \ell)}\log n$
and can be listed in time $2^\ell \ell^{O(\log \ell)} n\log n$.

Fix a cluster $U\in \Ufam$.
For every vertex $v\in U$, define its value with respect to $U$ by $\val_U(v) \coloneqq |N(v)\setminus (LN_\lambda(v)\cap U)|$,
and for $v\notin U$ put $\val_U(v)\coloneqq0$.

For this fixed set $U$, we consider the unary-weighted \textsc{\FO Value Discovery} instance with vertex value function $\val_U$.
There is no additional target constraint on the final tuple, except that the tokens occupy distinct vertices. The objective is to maximize $\sum_{v\in T}\val_U(v)$ over all sets $T\subseteq V(G)$ with $|T|=|S|=k$ that are reachable from $S$ within budget $b$.

By the discovery meta-theorem of Grobler et al.~\cite{GroblerMMMRSS24}, this unary-weighted discovery problem is fixed-parameter tractable with respect to $k$ whenever the associated weighted rainbow optimization problem is as well.
In this optimization problem, the task is to select a set $T$ of total weight $b$ and size $k$ that maximizes $\val_U(T)\coloneqq\sum_{v\in T}\val_U(v)$ and such that the vertices in~$T$ have distinct colors.
Here, the weighted rainbow optimization problem reduces immediately to \GK: vertices become items, colors become groups, the movement costs become weights, and values become profits. 
Standard dynamic programming solves \GK\ in
polynomial time~\cite{Pinedo16}. 
Hence, for each cluster $U\in \Ufam$, we can compute a reachable set $\mathrm{sol}(U)$ maximizing total value.

It remains to relate a value to the number of covered edges, as done by Panolan et al.~\cite{PanolanSST22}. 
For every cluster $U$ and every set $A\subseteq U$, we have
$
\sum_{v\in A}\val_U(v) \le \cov_G(A)$.
Indeed, each term $\val_U(v)$ counts neighbors of $v$ except those later than $v$ in the ordering $\lambda$ and contained in~$U$. 
Consider any edge with at least one endpoint in $A$. If both endpoints lie in $A$, then the edge is counted only from the later endpoint, since the earlier endpoint does not count later neighbors in~$U$. 
If exactly one endpoint lies in $A$, then the edge is counted at most once, namely from its endpoint in $A$, and possibly not counted if the other endpoint is a later neighbor in~$U$. 
Hence, every edge in $E(A)$ contributes at most one to the sum, proving the inequality.

Now let $I$ be a partial vertex cover of size $k$. Define $\widetilde I \coloneqq I \cup \bigcup_{v\in I} LN_\lambda(v)$.
Since every vertex has at most $d$ later neighbors, we have $|\widetilde I| \le k+kd \le \ell$.
The universal-set property applied to $A=\widetilde I$ implies that every subset of $\widetilde I$ appears as
$\widetilde I\cap U$ for some $U\in\mathcal U$. 
Applying this to the subset $I\subseteq \widetilde I$, there exists
$U\in\mathcal U$ such that $\widetilde I\cap U = I$.
Hence, $I\subseteq U$ and no later neighbor of a vertex in $I$ lies in
$U\setminus I$.
For this cluster,
$
\sum_{v\in I} \val_U(v)=\cov_G(I)$.
Consequently, if $(G,S,b,t)$ is a yes-instance, then for some cluster $U$ we obtain $\sum_{v\in \mathrm{sol}(U)} \val_U(v) \ge t$.

By the upper bound displayed above, this implies that $\mathrm{sol}(U)$ covers at least $t$ edges. 
Conversely, if some $\mathrm{sol}(U)$ covers at least $t$ edges, then it is a witness for the
original discovery instance. 
Altogether, iterating over all clusters gives a fixed-parameter tractable algorithm with running time $2^\ell \ell^{O(\log \ell)} n^{O(1)}$, where $\ell$ depends only on $k+d$.
\end{proof}

\subsection{Hardness on planar graphs and graphs of bounded cutwidth}
We now present a parameterized reduction from \textsc{Circulating Orientation}, defined below, to \textsc{Vertex Cover Discovery}, and prove the following. 

\begin{theorem}\label{thm:VC-Planar-Cutwidth}
\textsc{Vertex Cover Discovery} is \W[1]-hard with respect to parameter cutwidth and \NP-hard on planar graphs.
\end{theorem}

\begin{definition}[\textnormal{\textsc{Circulating Orientation}}]
Let $G=(V,E)$ be an undirected graph and let $w \colon E \to \mathbb{N}$ be an edge-weight function.
For an orientation $D$ of $G$, let $A(D)$ denote the corresponding set of directed arcs.
For a vertex $v\in V$, define the
\emph{weighted indegree} and \emph{weighted outdegree} of $v$ by
\[
    \indeg_w^D(v)\coloneqq\sum_{(u,v)\in A(D)} w(\{u,v\}) \qquad\text{ and }\qquad \out_w^D(v)\coloneqq\sum_{(v,u)\in A(D)} w(\{v,u\}) \text{, respectively.}
\]
A \emph{circulating orientation} of $(G,w)$ is an orientation $D$ of $G$
        such that
        \[
            \indeg_w^D(v)=\out_w^D(v)=\frac{1}{2}\sum_{e \ni v} w(e)
            \qquad\text{for every }v\in V.
        \]
The problem \textnormal{\textsc{Circulating Orientation}} asks whether a given weighted graph~$(G,w)$ admits a circulating orientation.
\end{definition}

\textsc{Circulating Orientation} is \W[1]-hard with respect to parameter pathwidth, even for planar instances and when every edge capacity is bounded by a polynomial in the input graph size~\cite{DBLP:conf/gd/JansenKKLMS23}.

We start the proof by describing the construction of a \textsc{Vertex Cover Discovery} instance when a \textsc{Circulating Orientation} instance $(G,w)$ is given.
Intuitively, we invert the direction of the problem.
Instead of viewing edges as contributing weight to their incident vertices, we decide, for each edge of $G$, which of its two endpoints must send tokens to the corresponding edge gadget.
The construction ensures that exactly one endpoint can send these tokens, thus the choice determines the orientation of the corresponding edge in the original \textsc{Circulating Orientation} instance. Figure~\ref{fig:co-reduction} shows an example of a source instance together with the corresponding reduced instance.

\begin{figure}
    \begin{tikzpicture}[baseline={(0,0)},mainvertex/.style={draw, circle, minimum height=1.7em},hiddenmainvertex/.style={minimum height=1.3em}, weight/.style={},
    pathvertex/.style={draw=darkgreen, fill=darkgreen, circle, minimum width=0.25em},token/.style={draw,rectangle, minimum height=0.25em,  minimum width=0.25em, fill=black, inner sep=0em},
    longpath/.style={dashed},
    edge/.style={draw},
    pathedge/.style={draw,very thick,darkgreen,overlay}
]
    \newcommand\yscale{0.6}
    \node[mainvertex,label=left:$v_1$] (v1) at (1,2.5*\yscale) {};

    \node[mainvertex, label=left:$v_2$] (v2) at (0,-2.5*\yscale) {};
    \node[mainvertex, label=left:$v_3$] (v3) at (2,-2.5*\yscale) {};

    \draw[edge] (v1) to node[left] {\textbf{\textcolor{darkgreen}{$2$}}} (v2);
    \draw[edge] (v1) to node[left] {\textbf{\textcolor{darkgreen}{$4$}}} node[right] {\textbf{\textcolor{darkgreen}{$e$}}} (v3);

\end{tikzpicture}
    \hspace{-20em}\hfill
    {\newcommand\token{\tikz{\node[token]{};}}
\begin{tikzpicture}[baseline={(0,0)},
    mainvertex/.style={draw, circle, minimum height=1.7em},
    shadowvertex/.style={draw=darkgray, circle, minimum width=0.25em},
    longpathvertex/.style={draw, circle, minimum width=0.25em},
    hiddenmainvertex/.style={minimum height=1.3em}, weight/.style={},
    pathvertex/.style={draw=darkgreen, thick, circle, minimum width=0.25em},
    onetokenvertex/.style={draw, thick, circle, minimum width=1.2em},
    simplevertex/.style={draw, thick, circle, minimum width=0.25em},
    token/.style={draw,rectangle, minimum height=0.25em,  minimum width=0.25em, fill=black, inner sep=0em},
    longpath/.style={dashed},
    shadowpath/.style={draw=darkgray},
    pathedge/.style={draw,thick,darkgreen,overlay},
    cut/.style={draw,thick,darkred,overlay},
    simplepath/.style={draw,thick,overlay},
    hidden/.style={draw=none}
]
    \newcommand\scalex{0.6}
    \newcommand\scaley{1}
    \newcommand\mainvertex[5][missing]{
        \node[mainvertex,label=left:#1] (v#2) at (#3*\scalex, #4*\scaley){};
        \tokens{v#2}{0.46em}{#5}
        \node[shadowvertex, label=right:\color{darkgray}#1$'$] (s#2) at (#3*\scalex+1.1*\scalex, #4*\scaley){};
        \draw[shadowpath] (v#2) -- (s#2);
    }
    \newcommand\activetoken[1]{
        \node[token] at (#1) {};
    }
    \newcommand\tokens[3]{
        \ifnum#3>0
            \foreach \tkn in {1, ..., \the\numexpr#3+1\relax} {
            \ifnum\tkn>1
            \node[token] (n\tkn) at ($(#1)+({(\tkn - 1) * 360 / (#3 +1)}:#2)$) {};
            \else
                \node[token,fill=darkgray,draw=darkgray] (n\tkn) at ($(#1)+({0}:#2)$) {};
            \fi
        }
        \fi
    }

    \newcommand{\connectmidpoints}[2][]{%
        \draw[smooth, tension=0.8, pathedge]  [#1]
          \foreach \first/\second [count=\i] in {#2} {
            \ifnum\i=1
                ($(\first)!0.5!(\second)$)
              \else
                  to ($(\first)!0.5!(\second)$)
              \fi
        };
    }

    \newcounter{xcoordc}
    \newcommand\nextx{\stepcounter{xcoordc}}
    \newcommand\xcoord{{\value{xcoordc}*\scalex}}
    \newcommand\ycoord{0}
    \newcommand{\accessnode}[1]{\the\numexpr (#1)/2\relax}
    \newcommand\gadget[5][]{
        \nextx

        \foreach \i in {1, ..., \the\numexpr#5*2\relax} {
            \nextx
            \node[pathvertex,#1] (w#2#3#4#2\i) at (\xcoord,\ycoord+\scaley) {};
            \node[pathvertex,#1] (w#2#3#4#3\i) at (\xcoord,\ycoord-\scaley) {};
       		\node[hidden] (x#2#3#4\i) at ({(\value{xcoordc}*\scalex + \scalex}, \ycoord+0.5*\scaley) {};
            \node[hidden] (y#2#3#4\i) at ({(\value{xcoordc}*\scalex + \scalex},\ycoord-0.5*\scaley) {};
            \node[hidden] (z#2#3#4\i) at ({(\value{xcoordc}*\scalex + 0.6*\scalex},\ycoord) {};
            \node[hidden] (s#2#3#4\i) at ({(\value{xcoordc}*\scalex + 1.4*\scalex},\ycoord) {};

        }
        \ifnum#5>1
        \foreach \i in {1, 3, ..., \the\numexpr#5*2-1\relax} {
            \newcommand\lastlongpathvertex{}
            
            \foreach \i in {2, 4, ..., \the\numexpr#5*2\relax} {
         		\node[onetokenvertex] (p#2#3#4\i) at (x#2#3#4\the\numexpr\i-1\relax) {};
         		\activetoken{p#2#3#4\i}
         		\node[onetokenvertex] (q#2#3#4\i) at (y#2#3#4\the\numexpr\i-1\relax) {};

                \activetoken{q#2#3#4\i}

                \ifnum\i>2

                \node[simplevertex] (p#2#3#4\the\numexpr\i-1\relax)  at (x#2#3#4\the\numexpr\i-2\relax) {};
                \node[simplevertex] (q#2#3#4\the\numexpr\i-1\relax)  at (y#2#3#4\the\numexpr\i-2\relax) {};
                \node[onetokenvertex] (r#2#3#4\the\numexpr\i-1\relax)  at (z#2#3#4\the\numexpr\i-2\relax) {};
                \node[shadowvertex] (rs#2#3#4\the\numexpr\i-1\relax)  at (s#2#3#4\the\numexpr\i-2\relax) {};

                \tokens{r#2#3#4\the\numexpr\i-1\relax}{0.26em}{1}

                \draw[simplepath] (p#2#3#4\i) -- (p#2#3#4\the\numexpr\i-1\relax);
                \draw[simplepath] (p#2#3#4\the\numexpr\i-1\relax) -- (p#2#3#4\the\numexpr\i-2\relax);
                \draw[simplepath] (q#2#3#4\i) -- (q#2#3#4\the\numexpr\i-1\relax);
                \draw[simplepath] (q#2#3#4\the\numexpr\i-1\relax) -- (q#2#3#4\the\numexpr\i-2\relax);
                \draw[simplepath] (r#2#3#4\the\numexpr\i-1\relax) -- (p#2#3#4\the\numexpr\i-1\relax);
                \draw[simplepath] (r#2#3#4\the\numexpr\i-1\relax) -- (q#2#3#4\the\numexpr\i-1\relax);
                \draw[simplepath] (r#2#3#4\the\numexpr\i-1\relax) -- (rs#2#3#4\the\numexpr\i-1\relax);

                \fi

                \draw[simplepath] (p#2#3#4\i) -- (q#2#3#4\i);
                \draw[simplepath] (p#2#3#4\i) -- (w#2#3#4#2\i);

                \draw[simplepath] (q#2#3#4\i) -- (w#2#3#4#3\i);
                
            }
        }
        \fi
            \draw[longpath] (v#2) -- (w#2#3#4#21);
			\draw[longpath] (v#3) -- (w#2#3#4#31);

        \foreach \i in {1, ..., \the\numexpr#5*2-1\relax} {
            \draw[pathedge] (w#2#3#4#2\i) -- (w#2#3#4#2\the\numexpr\i+1\relax);
            \draw[pathedge] (w#2#3#4#3\i) -- (w#2#3#4#3\the\numexpr\i+1\relax);
        }

        \draw[pathedge] (w#2#3#4#21) to (w#2#3#4#31);
        \draw[pathedge] (w#2#3#4#2\the\numexpr#5*2\relax) to [bend left=65,looseness=1.3] (w#2#3#4#3\the\numexpr#5*2\relax);
    }

    \mainvertex[$v_1$]{1}{4.5}{2.5}{3};
    \mainvertex[$v_2$]{2}{2}{-2.5}{1};
    \mainvertex[$v_3$]{3}{7}{-2.5}{2};

    \gadget{1}{2}{a}{2}
    \gadget{1}{3}{b}{4}

    \node[label=\color{darkgreen}$p_4^{{{e}}}$] at (w13b17) {};
    \node[label=\color{darkgreen}$\bar p_{4}^{{{e}}}$] at (w13b18) {};
    \node[label=right:$s_{4}^{{{e}}}$\vphantom{$t_4^e$}] at (x13b7) {};
    \node[label=right:$t_{4}^{{{e}}}$] at (y13b7) {};

    \node[label=below:\color{darkgreen}\vphantom{$d$}$q_{4}^{{{e}}}$] at (w13b37) {};
    \node[label=below:\color{darkgreen}$\bar q_{4}^{{{e}}}$] at (w13b38) {};

    \draw[smooth, tension=0.8, cut] plot coordinates {(1*\scalex, -1.5*\scaley) (3*\scalex, -1.5*\scaley)
    (3.5*\scalex, -0.6*\scaley) (4*\scalex, -0.1*\scaley) (4.15*\scalex, 0.2*\scaley) (4.4*\scalex, 0.5*\scaley) (4.5*\scalex, 1*\scaley) (4.5*\scalex, 1.5*\scaley)};
\end{tikzpicture}}
    \caption{\textsc{Circulating Orientation} instance with equivalent \textsc{Vertex Cover Discovery} instance. Dotted edges may be replaced by subdivided paths, as long as the gadget remains equidistant.
    In red, one of the worst case cuts of size $7$ that are required for an edge gadget is given.}
    \label{fig:co-reduction}
\end{figure}
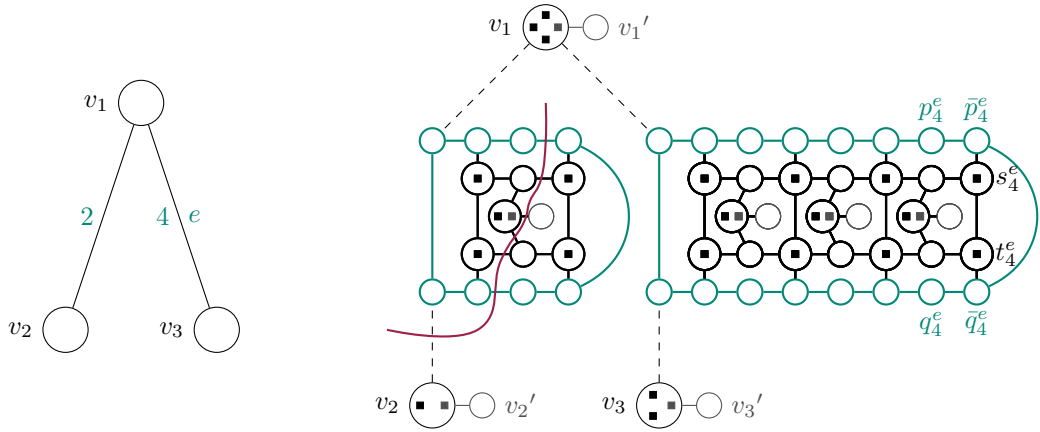

For a vertex $v_i\in V(G)$, let $W(v_i)\coloneqq\sum_{e \ni v_i} w(e)$ denote the total weight of its incident edges.
Additionally, we let $W\coloneqq\sum_{e\in E(G)} w(e)$.
If $W(v_i)$ is odd for some vertex~$v_i$, then the instance is trivially negative.
Edges of weight zero do not affect the constraints and may be deleted. After deleting such edges, isolated vertices may also be deleted. Hence, we may assume throughout that $W(v_i)$ is a positive even integer for every remaining vertex~$v_i$.
We construct an instance $(G',S,b)$ of \textsc{Vertex Cover Discovery} as follows.

We will argue later that the graph $G'$ can retain several properties, if present, of $G$.
Initially, we set $V(G') = V(G)$ and will subsequently add vertices to $G'$.
On every vertex $v_i\in V(G')$, we place $\frac{1}{2} W(v_i) + 1$ tokens.
For every vertex $v_i\in V(G')$, we add a shadow vertex $v_i'$ adjacent only to $v_i$; thus forcing at least one token to remain on $v_i$ to cover the edge ${v_i,v_i'}$.
For every edge $e = \edge{v_p}{v_q} \in E(G)$, we create an edge gadget $F_e$ as follows:
\begin{itemize}
    \item For every $i\in[w(e)]$, we create two choice vertices $p_i^e$ and $q_i^e$, two completion vertices $\bar p_i^e$ and $\bar q_i^e$, and two selector vertices $s_i^e$ and $t_i^e$. We write
    $P^e=\{p_i^e:i\in[w(e)]\}$, $Q^e=\{q_i^e:i\in[w(e)]\}$, and define $\bar P^e,\bar Q^e,S^e,T^e$ analogously. 
    In addition, for every $i\in[w(e)-1]$, we create a relay vertex $r_i^e$ and set $R^e=\{r_i^e:i\in[w(e)-1]\}$. 
    \item We connect the vertices of $P^e\cup \bar P^e$ and $Q^e\cup \bar Q^e$ alternately to form the two paths $\Pi_p^e=p_1^e,\bar p_1^e, p_2^e,\bar p_2^e, \dots, p_{w(e)}^e,\bar p_{w(e)}^e$ and $\Pi_q^e=q_1^e,\bar q_1^e, q_2^e, \bar q_2^e, \dots, q_{w(e)}^e, \bar q_{w(e)}^e$.
    \item We connect these two paths through the edges \edge{p_1^e}{q_1^e} and \edge{c_{w(e)}^e}{d_{w(e)}^e} such that they form a cycle containing $4w(e)$ vertices.
    These cycle vertices appear in green in Figure~\ref{fig:co-reduction}.
    \item For every $i \in [w(e)]$, we introduce the path $p_i^es_i^et_i^eq_i^e$ and place a token on both $s_i^e$ and $t_i^e$.
    \item For every $i \in [w(e)-1]$, we further introduce four vertices: the $i$th $p$-bridge vertex connected to $s_i^e$ and $s_{i+1}^e$, and the $i$th $q$-bridge vertex connected to $t_i^e$ and $t_{i+1}^e$.
    We also connect both vertices to the relay vertex $r_{i}^e$, which is only further connected to its own shadow vertex ${r_{i}^e}'$.
    \item We place two tokens on each $r_{i}$, with $i \in [w(e)-1]$, one of which is held in place by $r_{i}'$.
    \item We connect $p_1^e$ with $v_p$ and $q_1^e$ with $v_q$ and call the resulting edges attachment edges.
\end{itemize}

We next determine the movement budget contributed by the gadget $F_e$. The cycle of~$F_e$ will be covered in one of two symmetric ways. Either one token is moved from $v_p$ to each vertex of $P^e$, or one token is moved from $v_q$ to each vertex of $Q^e$.
Without loss of generality, consider the first case. By construction, the distances from $v_p$ to the vertices $p_1^e,p_2^e,\ldots,p_{w(e)}^e$
are $1,3,\ldots,2w(e)-1$.
Thus, moving one token from $v_p$ to each vertex of $P^e$ costs $w(e)^2$.
In addition, the gadget requires $w(e)$ one-step moves from either the selector set $S^e$ or the selector set $T^e$, and $w(e)-1$ one-step moves from the relay set $R^e$. Hence, the internal movement inside $F_e$ costs $w(e)+(w(e)-1)=2w(e)-1$.
Therefore, the total budget contribution of the edge gadget $F_e$ is $b_e \coloneqq w(e)^2+2w(e)-1$.

The only vertices of degree greater than four are the original vertices $v_i \in V(G')$.
We replace each with a binary caterpillar of depth $\deg_{G}(v_i)$.
This is an unbalanced binary tree with $v_i$ at the root, containing a single main branch on which all internal vertices of degree $3$ lie.
We place one token on each newly introduced internal vertex of a caterpillar and attach a shadow vertex to it.
For a caterpillar, each internal vertex is connected to a distinct edge gadget.
Suppose that the gadget $F_e$ is attached to the caterpillar for $v_p$ through an internal vertex at depth $\ell$, and to the caterpillar for $v_q$ through an internal vertex at depth $\ell' \le \ell$.
We then subdivide the attachment edge on the $v_q$-side $\ell-\ell'$ times.
After this subdivision, both endpoint reservoirs are at the same distance from the gadget.

We update the movement budget accordingly.
Every token that needs to move to the gadget $F_e$ needs to slide over each vertex of the caterpillar or edge subdivisions on its path to $v_p$ or $v_q$ once.
For a total distance $\ell$ to $v_p$ and $v_q$, we increase the budget by $\ell \cdot w(e)$.

We can now show that a vertex cover can be discovered if and only if a valid solution for \textsc{Circulating Orientation} exists.

Every edge, apart from edges in the cycles $C = \{C^e \mid e\in E(G)\}$, is covered initially.
Remaining to be covered are the disjoint cycles $C^e\in C$ with a total length of $4W$ vertices.
These cycles require at least $2W$ tokens to cover them (by having every other vertex with a token).
The internal structure of the edge gadgets has a surplus of $W$ tokens.
Using the tokens from $R^e$, at most $W$ tokens on the vertices of either $S^e$ or $T^e$ can be released (by moving them to bridge vertices). 
This consumes at least $2W-|E(G)|$ slides.
There are also exactly $W$ surplus tokens on the original vertices.
Therefore, all surplus $2W$ tokens must move to $C$ to obtain a full vertex cover.

For an edge gadget $F_e$, there are exactly two distinct final configurations for $C^e$: tokens are either on all of $P^e$ and $\bar Q^e$, and consequently, tokens are on all of $S^e$, or on all of $Q^e$ and~$\bar P^e$, and consequently, tokens are on all of $T^e$. 
The budget is sufficient to move tokens from either $v_p$ to the vertices of~$P^e$ or from $v_q$ to the vertices of $Q^e$, and perform the $2\cdot w(e) - 1$ slides to move both the tokens of either $S^e$ or $T^e$ and cover their thereby uncovered incident edges with the spare token from $R^e$.
As a result, for every gadget $F_e$, a solution can be found if all external tokens come either directly from $v_p$ or $v_q$.
It follows that every circulating orientation yields a feasible discovery sequence. 

For the other direction, we still need to rule out tokens moving from a vertex to an edge gadget of a non-incident edge in $G$ or an edge gadget $F_e$ for $e = \edge{v_p}{v_q}$ receiving tokens from both $v_p$ and $v_q$ in $G'$.
Both are ensured by having no spare tokens and no spare budget.
Without loss of generality, assume that $P^e$, $\bar Q^e$, and $S^e$ are full of tokens in the final configuration.
Any moving pattern besides the one described above will take an additional step, thus exceeding the total allowed budget of this gadget.
Similarly, any token coming from a different reservoir than $v_p$ or $v_q$ would have to traverse at least one additional edge.
Since all gadgets must be filled and there is no unused budget, such additional steps are impossible.
Hence, only one of the two reservoirs corresponding to the endpoints $v_p$ and $v_q$ can supply tokens to $F_e$.

For every edge $e=\edge{v_p}{v_q}$, the gadget $F_e$ has to absorb exactly $w(e)$ tokens from outside the gadget. By the budget-tightness argument, all these tokens must come either from the reservoir at $v_p$ or from the reservoir at $v_q$. We orient $e$ away from the endpoint that supplies these tokens.
For every vertex $v_i\in V(G)$, exactly $\frac12 W(v_i)$ of its reservoir tokens are allowed to move. Hence, the total weight of the edges oriented away from $v_i$ is exactly $\frac12 W(v_i)$. Since
\[
    W(v_i)=\sum_{e\ni v_i} w(e),
\]
the total weight of the edges oriented toward $v_i$ is also $\frac12 W(v_i)$. Thus, the orientation is circulating.

\begin{claim}\label{cl:planar}
If $G$ is planar, then so is $G'$.
\end{claim}

\begin{claimproof}
    Fix a planar embedding of $G$. We replace each vertex $v$ with a caterpillar embedded inside a sufficiently small disk around $v$, with the attachment points for the incident edge gadgets placed in the same cyclic order as the corresponding edges around $v$.
    Since caterpillars are trees, this can be done without introducing crossings.
    Each original edge is then replaced by a planar edge gadget embedded inside its thin neighborhood; the two attachment vertices of the gadget lie on the outer face of the gadget. Finally, all additional subdivisions and shadow leaves preserve planarity. Hence, the resulting graph $G'$ is planar.
\end{claimproof}

\begin{claim}\label{cl:pwgivescw}
If $G$ has bounded pathwidth, then there exists a construction of $G'$ that has bounded cutwidth.
\end{claim}

\begin{figure}\centering
    {\newcommand\token{\tikz{\node[token]{};}}

\newcommand\replacementedge[7][]{
    \draw (v#2#4) to (v#3#4);
}
\newcommand\cutlines{}

\newcommand\interval[4]{
    \foreach \i in {#3,...,#4} {
        \ifnum\i=#3
        \node[filledvertex,label={[label distance=0.4em,overlay]left:$v_{#1}$}] (v#1\i) at (\i,-#2) {};
        \else
            \node[filledvertex] (v#1\i) at (\i,-#2) {};
        \fi
    }

    \node[draw, rounded rectangle, fill=black,
        inner sep=0em,
        fit=(v#1#3) (v#1#4)] {};

}

\begin{tikzpicture}[baseline={(0,0)},
    mainvertex/.style={draw, circle, thick, minimum height=1.1em},
    longpathvertex/.style={draw, circle, minimum width=0.25em},
    hiddenmainvertex/.style={minimum height=1.3em}, weight/.style={},
    pathvertex/.style={draw=darkgreen, thick, circle, minimum width=0.25em},
    onetokenvertex/.style={draw, thick, circle, minimum width=1.2em},
    simplevertex/.style={draw, thick, circle, minimum width=0.3em, inner sep=0em},
    token/.style={draw,rectangle, minimum height=0.25em,  minimum width=0.25em, fill=black, inner sep=0em},
    longpath/.style={dashed},
    pathedge/.style={draw,thick,darkgreen,overlay},
    cut/.style={draw,thick,darkred,line cap=round},
    simplepath/.style={draw,ultra thick,overlay},
    hidden/.style={circle,draw=none,minimum width=0.5em,minimum height=0.5em,inner sep=0em},
    filledvertex/.style={draw,circle,fill=black,minimum width=0.5em,,minimum height=0.5em,inner sep=0em},
    smallfilledvertex/.style={draw,circle,fill=black,minimum width=0.4em,minimum height=0.4em,inner sep=0em},
    shadowvertex/.style={draw=darkgray,circle,fill=darkgray,minimum width=0.4em,minimum height=0.4em,inner sep=0em},
    shadowpath/.style={draw=darkgray},
]
\interval{a}{0}{0}{11}
\interval{b}{1}{0}{0}
\interval{c}{1}{2}{4}
\interval{d}{1}{9}{11}
\interval{e}{2}{4}{9}

\replacementedge{a}{b}{0}{0}{0}{0}
\replacementedge{a}{c}{2}{1}{0}{0}
\replacementedge{c}{e}{4}{1}{0}{0}
\replacementedge[minimum width=0em]{a}{e}{6}{1}{0}{0}
\replacementedge{d}{e}{9}{-2}{0}{0}
\replacementedge{a}{d}{11}{3}{0}{0}

\cutlines

\node[hidden] at (-1*2,0){};
\node[hidden] at (6*2,0){};

\end{tikzpicture}

}
    \caption{The layout of a \textsc{Circulating Orientation} instance in its interval representation}
    \label{fig:co-pathwidth}
\end{figure}
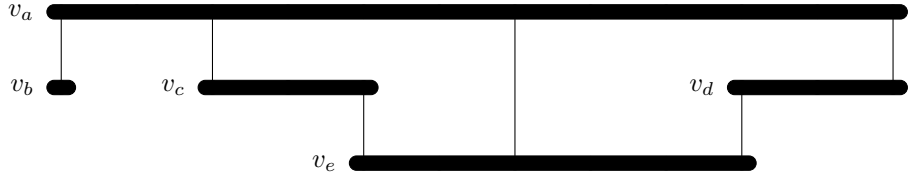

\begin{figure}\centering
    \input{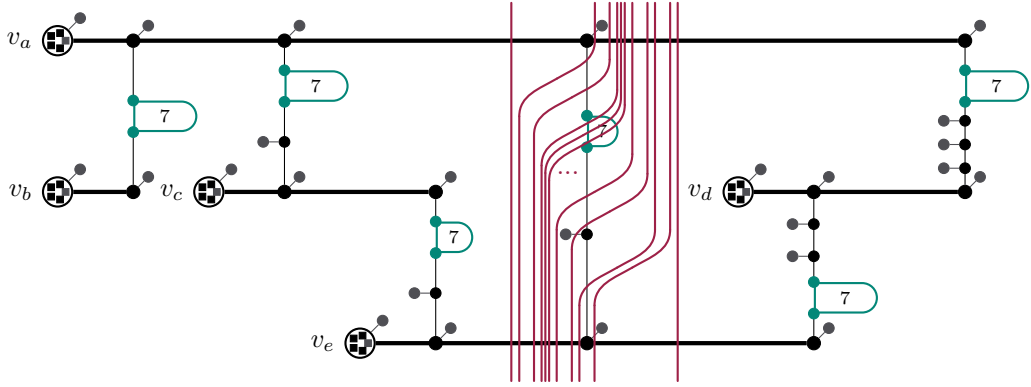}
    \caption{The structure of the \textsc{Vertex Cover Discovery} instance created from the interval representation in \cref{fig:co-pathwidth}.
    The cuts through the path representing an edge are shown in red}
    \label{fig:co-cutwidth}
\end{figure}

\begin{claimproof}
Let $\mathcal P=(B_1,\ldots,B_l)$ be a nice path decomposition of $G$ of width $\operatorname{pw}(G)$.
We can draw a graph $G$ with $\operatorname{pw}(G)-1$ overlapping intervals.
We provide an example of a graph in such an interval representation in~\cref{fig:co-pathwidth}.
We construct $G'$ with respect to $\mathcal P$ so that, for each vertex $v$, the attachment points on the caterpillar replacing $v$ appear in the left-to-right order in which the corresponding incident edges are encountered along $\mathcal P$.

Since every vertex of $G$ appears in a consecutive interval of bags, $\mathcal P$ induces an interval representation of the vertices of $G$.
To calculate the cutwidth of $G'$, this interval representation will be used.
Every vertex interval is replaced by the main branch of its caterpillar, with the original vertex, its shadow, and the first internal vertex placed consecutively, and thereafter each shadow leaf placed immediately after the corresponding caterpillar internal vertex.
At the end, internal vertices attached to the same edge gadget must be vertically aligned (i.e., longer edges can be drawn).

We define the cutwidth ordering by sweeping a vertical thin window through this template from left to right.
Whenever the window encounters original or internal vertices, we add them to the linear order in their top-to-bottom order along the sweep line, placing the shadow vertex of each such vertex immediately after it.
Note that at most one vertical edge between two caterpillars exists in any such window.
At any cutting position along the ordering thus far, the only edges from the caterpillars crossing the cut correspond to intervals of vertices that are active in the current bag $B_j$.
Hence, there are at most $|B_j|+1$ such crossing edges if the cut appears between an internal vertex and its shadow.
We give an example for the cuts in the interval between two edges in~\cref{fig:co-cutwidth} in red.
When a vertical path is encountered in the window, its vertices are added to the ordering between the two internal caterpillar vertices to which it attaches, following the shadow vertex corresponding to the first of these vertices.
A fixed internal ordering of the edge gadget contributes at most $7$ additional crossing edges on top of the edges of caterpillars.
An example for such a cut is shown in~\cref{fig:co-reduction} in red.
Therefore, every cut has a size of at most $|B_j|+7 \ge |B_j|+1$.
Since $|B_j|\leq \operatorname{pw}(G)+1$, the cutwidth of $G'$ is at most $\operatorname{pw}(G)+8$.
\end{claimproof}

Finally, we show that the construction takes polynomially many steps with respect to the size of $G$ whenever $W$ is bounded by a polynomial in $|V(G)|$.
Before introducing caterpillars, the construction adds $\mathcal O(W+|V(G)|+|E(G)|)$ vertices, and the caterpillars and subdivisions add only polynomially many further vertices.
The budget is polynomially bounded for the same reason.
Applying the reduction to the planar instances of \textsc{Circulating Orientation} that are \W[1]-hard parameterized by pathwidth and have polynomially bounded edge weights with~\cref{cl:planar}, we find that \textsc{Vertex Cover Discovery} is \NP-hard even on planar graphs. 
Given a path decomposition of width $p$ for $G$ (which can be computed in fixed-parameter tractable time with respect to $pw(G)$), the construction above with~\Cref{cl:pwgivescw} produces an equivalent \textsc{Vertex Cover Discovery} instance $G'$ of cutwidth at most $pw(G)+8$.
Therefore, \textsc{Vertex Cover Discovery} is \W[1]-hard parameterized by cutwidth.

\subsection{Hardness for clique cover number}

\begin{theorem}
\textsc{Vertex Cover Discovery} is $\W[1]$-hard when parameterized by the clique cover number of the input graph.
This remains true even if a clique cover of the claimed size is given as part of the input.
\end{theorem}

\begin{proof}
We reduce from \textsc{Multicolored Independent Set}, which is $\W[1]$-hard parameterized by the number $\kappa$ of color classes~\cite{CyganFKLMPPS15}.
Let $V_1 \cup \dots \cup V_\kappa$ be the color classes of the graph $G$, and the question is whether $G$ contains an independent set consisting of exactly one vertex from each $V_i$.

We construct an instance $(G',S,b)$ of \VCD. First, we make each color class $V_i$ induce a clique, and we keep the edges between distinct color classes exactly as in $G$. 
We then add new vertices $d_1,\ldots,d_\kappa$ and $u$, and let $D\coloneqq\{d_1,\ldots,d_\kappa\}$.
We make $D$ into a clique and make every vertex of $D$ adjacent to every vertex outside $D$.
Thus, each vertex of $D$ is universal. 
The vertex $u$ has no neighbors outside $D$.
We set $S \coloneqq V(G')\setminus (D \cup \{u\})$ and $b \coloneqq \kappa$.

The clique cover number of $G'$ is at most $\kappa+2$, since $V_1,\dots,V_\kappa,D, \{u\}$
is a clique cover of $G'$.
Hence, the parameter of the constructed instance is bounded by a function of $\kappa$.

\medskip
\noindent
\emph{($\Rightarrow$)} Assume that $G$ contains a partitioned independent set $I = \{v_1,\dots,v_\kappa\}$
with $v_i \in V_i$ for every $i \in [\kappa]$. Since every $d_i$ is adjacent to every vertex outside $D$, we may slide the token on $v_i$ to $d_i$ for every $i\in[\kappa]$. This uses exactly $\kappa=b$ slides. 
The token-free set after these moves is exactly $I \cup \{u\}$.
The set $I$ is independent in $G'$: it contains exactly one vertex from each color class, so no edge added inside a color class is used, and no two chosen vertices are adjacent between color classes because $I$ was independent in $G$. 
Additionally, $I \cup \{u\}$ is independent in $G'$.
Therefore, the final token set $V(G') \setminus I \setminus \{u\}$ is a vertex cover of $G'$, and $(G',S,b)$ is a yes-instance of \VCD.

\medskip

\noindent
\emph{($\Leftarrow$)} Assume now that $(G',S,b)$ is a yes-instance of \textsc{Vertex Cover Discovery}.
$D \cup \{u\}$ is the initial set of token free vertices.
Since a token set is a vertex cover if and only if its complement is an independent set, there is an independent set $H\subseteq V(G')$ of size $\kappa+1$ that is reachable from the initial token-free set $D \cup \{u\}$ within budget $\kappa$.

We first claim that in a solution that uses the minimum number of token slides, no vertex of $D$ belongs to $H$. 
Indeed, every vertex of $D$ is universal, while $|H|=\kappa+1$; for $\kappa+1\ge 2$, an independent set of size $\kappa+1$ cannot contain a universal vertex. 
Thus, all $\kappa$ vertices of $D$ must be filled by tokens during the movement. 
Therefore, $\kappa$ initial tokens move from $V_1\cup\dots\cup V_\kappa$ into $D$ using a budget of $\kappa$ slides.
Thus, $H \setminus \{u\} \subseteq V_1\cup\dots\cup V_\kappa$
and $|H\setminus \{u\}|=\kappa$.

Since each $V_i$ induces a clique in $G'$, the independent set $H \setminus \{u\}$ contains at most one vertex from each color class. As there are $\kappa$ color classes and $|H \setminus \{u\}|=\kappa$, it contains exactly one vertex from each $V_i$. Finally, the edges between distinct color classes in $G'$ are exactly the edges of $G$. Since $H\setminus \{u\}$ is independent in $G'$, it is a multicolored independent set in $G$.
Since the construction is polynomial-time computable, and the clique cover number of $G'$ is at most $\kappa+2$, the theorem follows.
\end{proof}

\subsection{Hardness for parameter $k+b$}
\label{subsec:pvc-hardness}

We also show that the parameter $k+b$ alone is not sufficient for the fixed-parameter tractability of \PVCD.

\begin{theorem}
\label{thm:pvcd-hard-new}
\PVCD is $\W[1]$-hard parameterized by $k+b$.
\end{theorem}

\begin{proof}
We give a parameterized reduction from the \textsc{Multicolored Independent Set} problem.
Let $(G,\chi)$ be an instance where every color class has size exactly $c > 1$ and $\Delta$ be the maximum degree in $G$.
From $(G,\chi)$, we construct a graph $G'$ as follows. For each color class, we add all missing edges inside the class, turning it into a clique, while keeping all original edges of $G$. For every original vertex $v\in V(G)$, we add $\Delta-\deg_G(v)$ private leaves adjacent only to $v$. Thus, every original vertex has degree exactly $(c-1)+\Delta$ in $G'$, while every private leaf has degree $1$.
Let $S$ be any set containing exactly one original vertex from each color class. 
We use $S$ as the starting token set, so $k=|S|=\kappa$, and set $b\coloneqq\kappa$ and $t\coloneqq(c-1+\Delta)\kappa$.

Assume first that $G$ has a multicolored independent set $I$. 
Since $I$ contains exactly one original vertex from each color class, and since every color class is a clique in $G'$, each token can be moved from its start vertex in $S$ to the corresponding vertex of $I$ within its color clique. 
This uses at most $\kappa=b$ slides.
Every chosen vertex covers exactly $c-1$ clique edges inside its color class and $\Delta$ other edges.
Since $I$ is independent in $G$, and by construction, no edge is counted twice, and exactly $(c-1+\Delta)\kappa$ edges are covered.

Conversely, suppose that there exists a reachable token set $C$ of size $k$ that covers at least $t=(c-1+\Delta)\kappa$ edges in $G'$. 
Every original vertex has degree $(c-1)+\Delta$, and every private leaf has degree $1$. 
Choosing any $\kappa$ vertices will cover at most $(c-1+\Delta)\kappa$ edges, with equality only if every vertex of $C$ is an original vertex.
Therefore, every vertex of $C$ is original, and no covered edge is counted twice; equivalently, no two vertices of $C$ are adjacent in $G'$.
Since each color class is a clique in $G'$, the set $C$ contains at most one vertex from each color class. As $|C|=\kappa$ and there are exactly $\kappa$ color classes, it contains exactly one vertex from each color class. Finally, all original edges of $G$ are present in $G'$, so the independence of $C$ in $G'$ implies that $C$ is independent in $G$. Thus, $C$ is a multicolored independent set of $G$.

The reduction is clearly computable in polynomial time, and the parameter $k+b$ depends only on $\kappa$.
Therefore, \PVCD\ is $\W[1]$-hard parameterized by~$k+b$.
\end{proof}

\section{FO Cost-Value Decision and Discovery}
\label{sec:fo-value-discovery}
We now turn to the logical part of the paper. We begin by fixing the logical notation and recalling the locality theorem used in the proof of the meta-theorem. 
We then present two small numerical reductions.

\paragraph*{First-order and monadic second-order logic.}
We use first-order logic \FO and monadic second-order logic \MSO over the usual graph language with equality and adjacency relation $E$, possibly expanded by finitely many unary predicates, called colors. We refer to~\cite{Libkin04} for background. 
First-order variables are interpreted as vertices, and monadic second-order variables as vertex sets. 
For a formula $\varphi(\bar x)$ and a tuple $\bar a$ of vertices, we write $G\models\varphi(\bar a)$ in the usual way. For fixed $r\in\N$, we use $\dist_{\le r}(x,y)$ for an \FO formula expressing $\dist_G(x,y)\le r$, and we write $\dist_{>r}(x,y)$ for its negation.

We use the following convenient free-variable form of Gaifman locality due to Grohe and Schweikardt~\cite{GroheSchweikardt26}. 
Although their theorem is stated for $\FO^+$ with fixed-radius distance atoms, we state it in ordinary \FO, since these atoms are first-order definable for fixed radius.

For $k\ge 1$, let $\mathcal H_k$ be the set of all graphs with vertex set $[k]$. For a graph $G$, a tuple $\bar a=(a_1,\ldots,a_k)\in V(G)^k$, and $r\in\N$, let $H_G^r(\bar a)\in\mathcal H_k$ have edge $\{i,j\}$ exactly if $\dist_G(a_i,a_j)\le r$. We write $\operatorname{cc}(H)$ for the connected components of $H$.

\begin{theorem}[Graph reformulation of Theorem~7.1 of~\cite{GroheSchweikardt26}]
\label{thm:refined-free-variable-gaifman}
Let $\varphi(x_1,\ldots,x_k)$ be an \FO formula. There is an effectively computable radius $r\in\N$ such that, for every $H\in\mathcal H_k$, one can effectively compute a finite index set $\Lambda_H$ and, for every $\tau\in\Lambda_H$, 
\begin{itemize}
    \item a tuple-independent sentence $\xi_H^\tau$, and
    \item for every $C\in\operatorname{cc}(H)$, an $r$-local formula
$\psi_{H,C}^\tau(\bar x_C)$
\end{itemize}
such that, for every graph $G$ and every tuple $\bar a\in V(G)^k$, $G\models\varphi(\bar a)$ if and only if for $H\coloneqq H_G^r(\bar a)$, there is a $\tau\in\Lambda_H$ such that $G\models\xi_H^\tau$ and $G\models\psi_{H,C}^\tau(\bar a_C)$ for every $C\in\operatorname{cc}(H)$.
Moreover, for fixed $G$ and $\bar a$, this index $\tau$ is unique whenever it exists.
\end{theorem}

The connected components of $H_G^r(\bar a)$ are the far-apart parts of the tuple. Indeed, if $C,D\in\operatorname{cc}(H_G^r(\bar a))$ are distinct, then $\dist_G(a_i,a_j)>r$ for all $i\in C$ and $j\in D$. Thus, after fixing the graph $H$ and the index $\tau$, the formula $\varphi$ is evaluated by one tuple-independent side condition and by local formulas on the far-apart components of $H$.

A unary expansion of a graph $G$ is obtained by adding finitely many unary predicates, interpreted as subsets of $V(G)$. 
A unary expansion of a graph class $\Cc$ is defined analogously, using a fixed finite set of additional unary predicates. Formulas over such expansions may use these predicates.

Value expressions are used only in the restricted sense of \Cref{def:weighted-fo-cost-value-optimization}: input weights are unary vertex weights, and the only non-unary contribution is the fixed constant attached to the logical case of the tuple. In particular, weights are not accessible to the logic.
\medskip

We now show the two numerical reductions.
The first reduction is a normalization of the cost constraint. 
We use it whenever an \FOCV instance is passed to one of the later optimization routines: it removes negative finite costs coordinatewise and replaces the original cost bound by a non-negative budget, while preserving exactly the tuples that satisfy the cost constraint. 
Thus all subsequent arguments may work with costs in $\N\cup\{+\infty\}$ and with a unary budget $B\in\N$, without repeating this numerical preprocessing.
The proof is immediate and is omitted.

\begin{observation}
\label{lem:focv-cost-normalization}
Let $G$ be a graph, let $c_1,\ldots,c_k\colon V(G)\to\Z_{+\infty}$ be cost functions, and let $C\in\Z$ be a cost bound.
In polynomial time, one can either correctly conclude that no tuple in~$V(G)^k$ has finite cost at most $C$, or compute cost functions $c'_1,\ldots,c'_k\colon V(G)\to\N\cup\{+\infty\}$ and a budget $B\in\N$ such that, for every tuple $\bar a=(a_1,\ldots,a_k)\in V(G)^k$, we have $\sum_{i=1}^k c_i(a_i)\le C$ if and only if $\sum_{i=1}^k c'_i(a_i)\le B$.
If all finite input numbers are encoded in unary, then polynomial dependence on $B$ is polynomial dependence on the input length.
\end{observation}



The second reduction prepares a fixed threshold case for an ordinary weighted tuple optimizer. 
After the costs have been normalized and a finite case of the value expression has been fixed, it caps all costs that cannot appear in a budget-$B$ solution by $B+1$ and replaces every value $-\infty$ by a finite sentinel that is too small to reach the current threshold. 
Again we omit the simple proof. 

\begin{observation}
\label{lem:finite-threshold-tuple-reduction}
Let $\theta(\bar x)$ be a first-order formula with $\bar x=(x_1,\ldots,x_k)$ and 
let $c_1,\ldots,c_k\colon$ $V(G)\to\N\cup\{+\infty\}$ be non-negative cost functions, let $w_1,\ldots,w_k\colon V(G)\to\Z_{-\infty}$ be value functions, let $B\in\N$, and let $T\in\Z$.

For every coordinate $i$, define a finite cost function $\gamma_i\colon V(G)\to\N$ by putting $\gamma_i(v)\coloneqq c_i(v)$ if $c_i(v)<+\infty$ and $c_i(v)\le B$, and putting $\gamma_i(v)\coloneqq B+1$ otherwise.

Let $W^+\coloneqq\sum_{i=1}^k\max(\{0\}\cup\{w_i(v):v\in V(G),\ w_i(v)>-\infty\})$ and $M\coloneqq W^++|T|+1$.
For every coordinate $i$, define a finite value function $\omega_i\colon V(G)\to\Z$ by putting $\omega_i(v)\coloneqq w_i(v)$ if $w_i(v)>-\infty$, and putting $\omega_i(v)\coloneqq-M$ otherwise.

Then, there is a tuple $\bar a$ such that $G\models\theta(\bar a)$, $\sum_i c_i(a_i)\le B$, and $\sum_i w_i(a_i)\ge T$ if and only if
\begin{equation*}
        \max\left\{\sum_{i=1}^k \omega_i(a_i):G\models\theta(\bar a),\ \sum_{i=1}^k\gamma_i(a_i)\le B\right\}\ge T.
\end{equation*}
The maximum is $-\infty$ if the feasible set is empty.
If all finite input numbers are encoded in unary, then all numbers introduced by the reduction have unary size polynomial in the input length.
\end{observation}





After normalizing the costs, the following tuple optimizer only has to handle one formula at a time. 

\begin{definition}[Efficient weighted \FO tuple optimization]
\label{def:weighted-fo-tuple-optimization}
A graph class $\Cc$ admits \emph{efficient weighted \FO tuple optimization} if, for every fixed colored first-order formula $\theta(\bar x)$, where $\bar x=(x_1,\ldots,x_k)$, the following problem is fixed-parameter tractable parameterized by $k$ and~$\theta$.

The input consists of a graph $G\in\Cc$, finite value functions $w_1,\ldots,w_k\colon V(G)\to\Z$, finite cost functions $\gamma_1,\ldots,\gamma_k\colon V(G)\to\N$, and a budget $B\in\N$. All numerical inputs are encoded in unary.
The task is to compute
\begin{equation*}
        \max\left\{\sum_{i=1}^k w_i(a_i):G\models\theta(\bar a),\ \sum_{i=1}^k\gamma_i(a_i)\le B\right\},
\end{equation*}
with value $-\infty$ if the feasible set is empty.
\end{definition}

An \FOCV instance with cases $(\varphi_i,\alpha_i)_{i=1}^m$ is then handled by calling it separately for each case.
We again omit the proof of the following simple lemma. 

\begin{lemma}
\label{lem:tuple-optimization-implies-focv}
If a graph class $\Cc$ admits efficient weighted \FO tuple optimization, then \FOCV is fixed-parameter tractable on $\Cc$, parameterized by the arity and by the defining value expression.
\end{lemma}





\begin{lemma}
\label{lem:local-neighborhood-optimization}
Let $\Cc$ be a graph class with fixed-parameter tractable first-order model checking. 
Assume that, for every radius $s\in\N$, the class of all induced neighborhoods $G[N_s^G[a]]$ with $G\in\Cc$ and $a\in V(G)$ admits efficient weighted \FO tuple optimization, also after adding a unary predicate for the center.
Then $\Cc$ admits efficient \textsc{Local} \FOCV.
\end{lemma}

\begin{proof}
Fix the radius $s$ of the local instance, the arity $k$, and the centered value expression. 
Let the input be a graph $G\in\Cc$, a center $a\in V(G)$, cost functions, value functions, a cost bound $C$, and a value threshold $W$.
Apply \Cref{lem:focv-cost-normalization} to the cost functions.
If the normalization rejects, then no local tuple can satisfy the cost bound, and we reject.
We keep the notation $c_i$ for the normalized costs and write $B$ for the normalized budget.

Consider one finite case $(\varphi(\bar x,z),\alpha)$ of the centered value expression, and put $T\coloneqq W-\alpha$.
Apply the refined free-variable Gaifman normal form to the formula $\varphi(\bar x,z)$, treating $z$ as an additional free variable.
Let $r$ be the resulting locality radius. For every normal-form alternative, we test its tuple-independent sentence in $G$ using the assumed first-order model-checking algorithm.
If the side sentence is false, we discard the alternative.

Let $G_a\coloneqq G[N_{s+r}^G[a]]$, expanded by a unary predicate $Z$ that marks the single vertex~$a$.
All vertices assigned to $\bar x$ in the local problem lie in $N_s^G[a]$.
Therefore, every radius-$r$ neighborhood of any free variable from $\bar x,z$ is contained in $G_a$.
Consequently, the distance pattern on $\bar x,z$ and all local formulas occurring in the normal form have the same truth value in $G$ and in the induced colored graph $G_a$.

For each surviving normal-form alternative, we form a colored first-order formula $\theta(\bar x)$ over~$G_a$. 
It has the form $\exists z\,(Z(z)\wedge\theta_0(\bar x,z))$, where $\theta_0$ states that every coordinate $x_i$ lies within a distance of at most $s$ from $z$, that the prescribed radius-$r$ distance pattern on $\bar x,z$ holds, and that all local formulas of the alternative hold.
We apply \Cref{lem:finite-threshold-tuple-reduction} on $G_a$ to this formula, to the restrictions of the normalized costs and values, to the budget $B$, and to the threshold~$T$.

The weighted tuple optimizer on the neighborhood class computes the corresponding finite optimum.
By \Cref{lem:finite-threshold-tuple-reduction}, the optimum reaches $T$ if and only if the local \FOCV instance has a tuple satisfying this case and this normal-form alternative.
Taking the disjunction over all finite cases and all surviving alternatives decides the local instance.
The number of alternatives and the radii depend only on the fixed local radius and on the value expression.
Thus, the running time is fixed-parameter tractable.
\end{proof}

We now prove the main theorem. 

\conditionalfovaluediscovery*

\begin{proof}
Let an instance of \FOCV be given.
Let its arity be $k$, let $G\in\Cc$ be the input graph, let $c_1,\ldots,c_k\colon V(G)\to\Z_{+\infty}$ be the cost functions, and let $w_1,\ldots,w_k\colon V(G)\to\Z_{-\infty}$ be the value functions.
Let the value expression be given by formulas $\varphi_1(\bar x),\ldots,\varphi_m(\bar x)$, where $\bar x=(x_1,\ldots,x_k)$, and constants $\alpha_1,\ldots,\alpha_m\in\Z_{-\infty}$.
Let $C$ be the cost bound, and let $W$ be the value threshold.
Apply \Cref{lem:focv-cost-normalization} to the cost functions.
If the normalization rejects, then no tuple of finite cost at most $C$ exists, and we reject.
We keep the notation $c_i$ for the normalized costs and write $B$ for the normalized budget.

We consider the cases of the value expression separately.
Cases with $\alpha_j=-\infty$ cannot yield a tuple of finite value at least the finite threshold $W$, and we ignore them.
Fix a remaining case $j$ and put $T_j\coloneqq W-\alpha_j$.
For this case, we have to decide whether there is a tuple $\bar a=(a_1,\ldots,a_k)$ such that $G\models\varphi_j(\bar a)$, $\sum_{i=1}^k c_i(a_i)\le B$, and $\sum_{i=1}^k w_i(a_i)\ge T_j$.

We apply the refined free-variable Gaifman normal form from \Cref{thm:refined-free-variable-gaifman} to the formula~$\varphi_j(\bar x)$.
This yields a radius $r$ and, for every graph $H\in\mathcal H_k$, a finite index set~$\Lambda_H$.
For every $\tau\in\Lambda_H$, it yields a tuple-independent side condition $\xi_H^\tau$ and, for every connected component $D$ of $H$, an $r$-local formula $\psi_{H,D}^\tau(\bar x_D)$.
All these objects depend only on the fixed formula $\varphi_j$.

We fix $H\in\mathcal H_k$ and $\tau\in\Lambda_H$.
Let $\mathcal P(H)=\{B_1,\ldots,B_p\}$ be the set of connected components of $H$.
We first test whether $G\models\xi_H^\tau$, which we can efficiently do by assumption.
If $G\not\models\xi_H^\tau$, then this choice of $H$ and $\tau$ is discarded.
Otherwise, $H$ describes the radius-$r$ distance pattern of the whole tuple, while the formulas $\psi_{H,B_h}^\tau$ describe the local behavior inside the components of this pattern.

For $I\subseteq [k]$, let $\delta_{H[I],r}(\bar x_I)$ be the first-order formula expressing the radius-$r$ distance pattern induced by $H$ on $I$, namely
\begin{align*}
        \delta_{H[I],r}(\bar x_I)
        \coloneqq
        \bigwedge_{\substack{u<v\\ u,v\in I\\ \{u,v\}\in E(H)}} \dist_{\le r}(x_u,x_v)
        \wedge
        \bigwedge_{\substack{u<v\\ u,v\in I\\ \{u,v\}\notin E(H)}} \dist_{>r}(x_u,x_v).
\end{align*}
Thus, $G\models\delta_{H,r}(\bar a)$ holds if and only if $H=H_G^r(\bar a)$.
For every component $B_h$, let $\rho_h\coloneqq\min B_h$.
Put $\widehat r\coloneqq\max\left\{1,r\right\}$ and $r^\ast\coloneqq(k+1)\widehat r$.
Since $B_h$ is connected in $H$, the formula $\delta_{H[B_h],r}(\bar x_{B_h})$ implies that every variable $x_i$ with $i\in B_h$ lies within a distance of at most $(|B_h|-1)r\le r^\ast$ from~$x_{\rho_h}$.

We now introduce a bounded coarsening of the Gaifman components.
Set $\sigma_0\coloneqq0$.
For $\ell\in\{0,\ldots,k\}$, define $D_\ell\coloneqq(k-1)\sigma_\ell+r^\ast$ and $R_\ell\coloneqq2D_\ell+r$.
For $\ell<k$, define $\sigma_{\ell+1}\coloneqq R_\ell$.
These numbers depend only on $k$ and $\varphi_j$.
The intended meaning is that components at distance at most $\sigma_\ell$ are grouped into one cluster, every such cluster is contained in a ball of radius $D_\ell$ around a cluster anchor, and distinct clusters are separated by distance greater than $R_\ell$ between their anchors.

We use the following elementary observation.
Let $\bar a$ be any tuple with $H_G^r(\bar a)=H$, and let $d_h\coloneqq a_{\rho_h}$ for $h\in[p]$.
For a scale $\sigma$, let $Q_\sigma(\bar a)$ be the graph on vertex set $[p]$ in which~$hh'$ is an edge if and only if $\dist_G(d_h,d_{h'})\le\sigma$.
As $\ell$ increases, connected components of $Q_{\sigma_\ell}(\bar a)$ can only merge.
If, for every $\ell\in\{0,\ldots,p-1\}$, the graphs $Q_{\sigma_\ell}(\bar a)$ and $Q_{R_\ell}(\bar a)$ had different connected components, then the number of connected components would decrease by at least $p$ times.
This is impossible.
Hence, there is an $\ell\in\{0,\ldots,p-1\}$ such that $Q_{\sigma_\ell}(\bar a)$ and $Q_{R_\ell}(\bar a)$ have the same connected components.
For this value of $\ell$, every component of $Q_{\sigma_\ell}(\bar a)$ has a representative whose whole cluster lies in its radius-$D_\ell$ neighborhood, and representatives of distinct clusters are at a distance greater than $R_\ell$.

We enumerate this coarsening.
Fix $\ell\in\{0,\ldots,p-1\}$ and a forest $F$ on the vertex set $[p]$.
Let $\mathcal Q(F)=\{C_1,\ldots,C_q\}$ be the set of connected components of $F$.
For a cluster $C\in\mathcal Q(F)$, put $I_C\coloneqq\bigcup_{h\in C}B_h$ and let $\eta_C\coloneqq\min C$.
The forest $F[C]$ is meant to witness that the representatives of the blocks in $C$ are connected at scale $\sigma_\ell$.

For a cluster $C\in\mathcal Q(F)$, define the centered cluster formula $\chi_{H,\tau,\ell,F,C}(\bar x_{I_C},z_C)$ by
\begin{align*}
        \chi_{H,\tau,\ell,F,C}(\bar x_{I_C},z_C)
        \coloneqq
        &(z_C=x_{\rho_{\eta_C}})
        \wedge
        \delta_{H[I_C],r}(\bar x_{I_C})\\
        &\wedge
        \bigwedge_{h\in C}\psi_{H,B_h}^\tau(\bar x_{B_h})\\
        &\wedge
        \bigwedge_{hh'\in E(F[C])}\dist_{\le\sigma_\ell}(x_{\rho_h},x_{\rho_{h'}})\\
        &\wedge
        \bigwedge_{i\in I_C}\dist_{\le D_\ell}(z_C,x_i).
\end{align*}
The last line is a protection condition.
It ensures that the whole partial tuple of the cluster is searched within the radius-$D_\ell$ ball around its anchor.
For fixed $k$ and $\varphi_j$, the formula $\chi_{H,\tau,\ell,F,C}$ has a fixed size.

For every cluster $C\in\mathcal Q(F)$, every center $a\in V(G)$, and every local budget $q'\in\{0,\ldots,B\}$, we solve the corresponding \textsc{Local} \FOCV instance.
Its arity is~$|I_C|$, its cost functions are $c_i$ for $i\in I_C$, its value functions are $w_i$ for $i\in I_C$, its center is $a$, its radius is $D_\ell+r$, and its centered value expression has the two cases $\chi_{H,\tau,\ell,F,C}$ with correction~$0$ and $\neg\chi_{H,\tau,\ell,F,C}$ with correction $-\infty$.
From the local decision algorithm, we compute the finite optimum by enumerating the unary range of possible finite values.
For every $i\in I_C$, let $m_i^-\coloneqq\min(\{0\}\cup\{w_i(v):w_i(v)>-\infty\})$ and $m_i^+\coloneqq\max(\{0\}\cup\{w_i(v):w_i(v)>-\infty\})$.
Every finite local value lies in the interval $[\sum_{i\in I_C}m_i^-,\sum_{i\in I_C}m_i^+]$, and this interval has unary-polynomial length.
Let this optimum be
\begin{align*}
        \nu_{H,\tau,\ell,F,C}(a,q')
        \coloneqq
        \max\Big\{
        \sum_{i\in I_C} w_i(x_i):
        &\ G\models\chi_{H,\tau,\ell,F,C}(\bar x_{I_C},a),\\
        &\ \bar x_{I_C}\in N_{D_\ell}^G[a]^{I_C},
        \sum_{i\in I_C}c_i(x_i)\le q'
        \Big\}.
\end{align*}
If no such local realization has a finite value, then we do not create a candidate for this choice.
There are at most $|V(G)|(B+1)$ local optimization computations for each cluster.
Since $B$ is unary, this is polynomial in the input size for fixed $k$ and fixed value expression.

We build an anchored weighted multicolored distance-$R_\ell$ independence instance.
The colors are the clusters in $\mathcal Q(F)$.
For a color $C$, the candidate set $A_C$ consists of all triples $y=(C,a,q')$ with $a\in V(G)$, $q'\in\{0,\ldots,B\}$, and $\nu_{H,\tau,\ell,F,C}(a,q')>-\infty$.
The anchor of~$y$ is $a(y)\coloneqq a$; its cost is $\kappa(y)\coloneqq q'$ and its profit is $\nu(y)\coloneqq\nu_{H,\tau,\ell,F,C}(a,q')$.
The budget is~$B$, and the profit threshold is $T_j$.
The budget is bounded by the number of budget states represented in the candidate sets and is therefore polynomially bounded in the size of the auxiliary instance.

We prove soundness.
Suppose that the auxiliary instance has a feasible solution, and let $y_C=(C,a_C,q_C)$ be the candidate chosen for cluster $C$.
By the definition of $\nu_{H,\tau,\ell,F,C}(a_C,q_C)$, there is a local tuple $\bar b_C=(b_i)_{i\in I_C}$ satisfying the cluster formula, with local cost at most $q_C$ and local value $\nu_{H,\tau,\ell,F,C}(a_C,q_C)$.
Combining the local tuples gives a tuple $\bar b=(b_1,\ldots,b_k)$.
For every cluster $C$ and every $i\in I_C$, the protection condition gives $\dist_G(a_C,b_i)\le D_\ell$.
Since anchors of distinct clusters are at a distance greater than $R_\ell=2D_\ell+r$, it follows that $\dist_G(b_i,b_{i'})>r$ whenever $i$ and $i'$ belong to distinct clusters.
Inside each cluster, the formula~$\delta_{H[I_C],r}$ enforces the radius-$r$ pattern induced by $H$.
Thus, $H=H_G^r(\bar b)$.
Moreover, every local formula $\psi_{H,B_h}^\tau(\bar b_{B_h})$ holds.
Together with $G\models\xi_H^\tau$, the refined Gaifman normal form implies $G\models\varphi_j(\bar b)$.
The total cost is at most $\sum_{C\in\mathcal Q(F)}q_C\le B$.
The total value satisfies
\begin{align*}
        \operatorname{val}(\bar b)
        =
        \sum_{i=1}^k w_i(b_i)+\alpha_j
        =
        \sum_{C\in\mathcal Q(F)}\nu_{H,\tau,\ell,F,C}(a_C,q_C)+\alpha_j
        \ge
        T_j+\alpha_j
        =
        W.
\end{align*}
Hence, $\bar b$ is feasible for the normalized instance, and therefore also for the original instance by \Cref{lem:focv-cost-normalization}.

We prove completeness.
Suppose that the original instance has a feasible tuple $\bar b$ whose value is realized by the case $j$.
After the normalization of \Cref{lem:focv-cost-normalization}, the same tuple has a normalized cost of at most $B$.
Let $H\coloneqq H_G^r(\bar b)$.
By the refined Gaifman normal form, there is a unique $\tau\in\Lambda_H$ such that $G\models\xi_H^\tau$ and $G\models\psi_{H,B_h}^\tau(\bar b_{B_h})$ hold for all components~$B_h$ of $H$.
For each $h\in[p]$, put $d_h\coloneqq b_{\rho_h}$.
By the scale observation, choose $\ell\in\{0,\ldots,p-1\}$ such that $Q_{\sigma_\ell}(\bar b)$ and $Q_{R_\ell}(\bar b)$ have the same connected components.
Choose a forest $F$ whose components are these common components and whose edges all lie in~$Q_{\sigma_\ell}(\bar b)$.
Let~$\mathcal Q(F)$ be the resulting cluster partition.
For every cluster $C\in\mathcal Q(F)$, set $a_C\coloneqq b_{\rho_{\eta_C}}$ and $q_C\coloneqq\sum_{i\in I_C}c_i(b_i)$, using the normalized costs.
The edges of $F[C]$ show that the block representatives in $C$ are connected at scale $\sigma_\ell$.
Every vertex of a block $B_h$ lies within a distance of at most $r^\ast$ from its representative $b_{\rho_h}$.
Thus, every vertex $b_i$ with $i\in I_C$ lies within a distance of at most~$D_\ell$ of~$a_C$.
Consequently, $G\models\chi_{H,\tau,\ell,F,C}(\bar b_{I_C},a_C)$.
The candidate $(C,a_C,q_C)$ is therefore present, and its profit is at least $\sum_{i\in I_C}w_i(b_i)$.
If~$C$ and $C'$ are distinct clusters, then they are distinct connected components of $Q_{R_\ell}(\bar b)$, so $\dist_G(a_C,a_{C'})>R_\ell$.
The candidate costs sum to the normalized cost of $\bar b$, which is at most~$B$.
The candidate profits sum to at least $\sum_{i=1}^k w_i(b_i)\ge W-\alpha_j=T_j$.
Thus, the auxiliary anchored distance-independence instance is feasible.

We run the construction for all finite cases $j$, all graphs $H\in\mathcal H_k$, all indices $\tau\in\Lambda_H$, all scale indices $\ell\in\{0,\ldots,k-1\}$, and all forests $F$ on the component set of $H$.
The number of such choices depends only on $k$ and the defining value expression.
Each auxiliary instance has at most $k$ colors and a distance parameter depending only on $k$ and the value expression.
By assumption, each auxiliary instance is solved in fixed-parameter time.
The original instance is accepted if and only if at least one auxiliary instance is feasible.
This proves fixed-parameter tractability of \FOCV on $\Cc$.
\end{proof}

\subsection{(Locally) structurally bounded expansion}
We now instantiate the general reduction for graph classes of locally structurally bounded expansion. We use structurally bounded expansion in the standard sense of Gajarsk\'y et al.~\cite{DBLP:journals/tocl/GajarskyKNMPST20}: a graph class has structurally bounded expansion if it is first-order transducible from a class of bounded expansion. A graph class has locally structurally bounded expansion if, for every $r\in\N$, the class of all induced neighborhoods $G[N_r^G[v]]$, with $G\in\Cc$ and $v\in V(G)$, has structurally bounded expansion.

We first prove the weighted tuple optimizer for structurally bounded expansion classes.
The proof uses the efficient reversal theorem for transductions of sparse graph classes due to Dreier, Gajarsk\'y, and Pilipczuk~\cite{DreierGP26}.
We use the following form of their result: for every graph class $\Cc$ of structurally bounded expansion, there is a polynomial-time algorithm that, given $G\in\Cc$, computes a vertex-colored graph $H$ from a bounded-expansion class and a fixed one-dimensional first-order interpretation $I$ such that $G=I(H)$.
We use a unary predicate~$U$ in $H$ for the domain of the interpretation, so all variables of pulled-back formulas are relativized to $U$.

We also use the aggregate-query machinery for bounded-expansion classes due to Toru\'nczyk~\cite{TorunczykPODS2020}.
In the max-plus form needed here, for every fixed first-order formula over a bounded-expansion class and unary numerical annotations on the vertices, one can compute the maximum annotated sum over all satisfying tuples in time $f(\varphi)\cdot |V(H)|^{O(1)}$, with arithmetic polynomial in the bit length of the annotations.

The following is the standard interpretation lemma, see~\cite{grohe2009methods}.

\begin{lemma}
\label{lem:pullback-through-interpretation}
Let $I$ be a fixed one-dimensional first-order interpretation with domain formula $\delta(x)$ and edge formula $\eta(x,y)$.
For every first-order formula $\varphi(\bar x)$ over graphs, one can effectively compute a first-order formula $\varphi^I(\bar x)$ over the source structure such that, whenever $G=I(H)$ and $\bar a$ is a tuple of vertices satisfying the domain formula, we have $G\models\varphi(\bar a)$ if and only if $H\models\varphi^I(\bar a)$.
Moreover, $\varphi^I$ can be chosen so that all quantified variables are relativized to $\delta$.
\end{lemma}

\begin{lemma}
\label{lem:sbe-weighted-fo-cost-value-optimization}
Every graph class of structurally bounded expansion admits efficient weighted \FO tuple optimization.
\end{lemma}

\begin{proof}
Let $\Cc$ be a graph class of structurally bounded expansion.
Fix a colored first-order formula $\varphi(\bar x)$, where $\bar x=(x_1,\ldots,x_k)$.
Let $G\in\Cc$ be given with finite value functions $w_i\colon V(G)\to\Z$ and finite cost functions $\gamma_i\colon V(G)\to\N$, all encoded in unary, and with a unary budget $B$.

By the efficient reversal theorem, we compute in polynomial time a vertex-colored graph~$H$ from a bounded-expansion class and a fixed one-dimensional interpretation $I$ such that $G=I(H)$.
We assume that the domain of the interpretation is marked by a unary predicate $U$, and we transfer the value and cost functions to the corresponding vertices of $H$ satisfying $U$.

By \Cref{lem:pullback-through-interpretation}, we compute a formula $\varphi^I(\bar x)$ over $H$ such that $G\models\varphi(\bar a)$ if and only if $H\models\varphi^I(\bar a)$ for all tuples $\bar a$ from $U$.
The required optimum is
$\max\big\{\sum_{i=1}^k w_i(x_i):H\models\varphi^I(\bar x)$, $\sum_{i=1}^k\gamma_i(x_i)\le B\big\}$.
This is a max-plus aggregate query with a unary additive cost constraint on a bounded-expansion class.
By the aggregate-query evaluation theorem, the optimum is computed in time $f(\varphi,k)\cdot |V(H)|^{O(1)}$.

The graph $H$ has size polynomial in $|V(G)|$ by the reversal theorem.
The arithmetic consists only of additions and comparisons of the input annotations.
\end{proof}

\begin{corollary}
\label{lem:locally-sbe-local-focv}
Every graph class of locally structurally bounded expansion admits efficient \textsc{Local} \FOCV.
\end{corollary}

\begin{proof}
Let $\Cc$ be a graph class of locally structurally bounded expansion.
For every fixed radius $s$, the class of all induced neighborhoods $G[N_s^G[a]]$ with $G\in\Cc$ and $a\in V(G)$ has structurally bounded expansion.
By \Cref{lem:sbe-weighted-fo-cost-value-optimization}, every such neighborhood class admits efficient weighted \FO tuple optimization, even after adding finitely many unary predicates.

First-order model checking on locally structurally bounded expansion classes follows from Gaifman locality, together with the same bounded-neighborhood reversal and bounded expansion model checking machinery.
Thus, the hypotheses of \Cref{lem:local-neighborhood-optimization} are satisfied, and the claim follows.
\end{proof}

\begin{lemma}
\label{lem:locally-sbe-anchored-wmdi}
Every graph class of locally structurally bounded expansion admits efficient anchored weighted multicolored distance-$r$ independence for every fixed $r\in\N$.
\end{lemma}

\begin{proof}
Classes of locally structurally bounded expansion are monadically stable.
Hence, the claim follows from the anchored distance-independence lemma for monadically stable classes, \Cref{lem:monadic-stable-anchored-wmdi}, proved below.
\end{proof}

\sbefocv*

\begin{proof}
By \Cref{lem:locally-sbe-local-focv}, the class $\Cc$ admits efficient \textsc{Local} \FOCV.
By \Cref{lem:locally-sbe-anchored-wmdi}, it admits efficient anchored weighted multicolored distance-$r$ independence for every fixed $r\in\N$.
The result follows from \Cref{thm:conditional-FO-value-discovery}.
\end{proof}

\subsection{(Locally) bounded cliquewidth}

We next prove the results for bounded and locally bounded cliquewidth.
For the global theorem, we use the LinEMSOL optimization theorem of Courcelle, Makowsky, and Rotics~\cite{CourcelleMR00}. A cliquewidth expression of width bounded by a function of $q$ whenever the cliquewidth of $G$ is at most $q$ can be computed in time $f(q)\cdot |V(G)|^{O(1)}$ by the algorithm of 
Oum and Seymour~\cite{OumSeymour06}.

\begin{lemma}
\label{lem:cw-weighted-fo-cost-value-optimization}
For every $q\in\N$, the class of graphs of cliquewidth at most $q$ admits efficient weighted \FO tuple optimization.
\end{lemma}

\begin{proof}
Fix a colored first-order formula $\varphi(\bar x)$, where $\bar x=(x_1,\ldots,x_k)$.
Let $G$ be a graph of cliquewidth at most $q$, let $w_i\colon V(G)\to\Z$ be finite value functions, let $\gamma_i\colon V(G)\to\N$ be finite cost functions, all encoded in unary, and let $B\in\N$ be a unary budget.
We compute a cliquewidth expression of width bounded by a function of $q$ using the approximation algorithm of Oum and Seymour.
Additional unary colors can be incorporated by multiplying the number of labels by a constant depending only on the fixed formula.

We translate the tuple variables into singleton set variables $X_1,\ldots,X_k$ in the standard way.
We obtain an $\MSO_1$ formula $\varphi^\star(X_1,\ldots,X_k)$ such that singleton assignments to the sets correspond exactly to tuples satisfying $\varphi$.

We optimize the linear objective $\sum_{i=1}^k\sum_{v\in X_i}w_i(v)$ subject to the linear side constraint $\sum_{i=1}^k\sum_{v\in X_i}\gamma_i(v)\le B$.
By the LinEMSOL optimization theorem on graphs of bounded cliquewidth, this optimization problem is solvable in time $f(q,k,\varphi)\cdot |V(G)|^{O(1)}$, with arithmetic polynomial in the bit length of the input values and in the unary budget.
This yields exactly the required weighted tuple optimum.
\end{proof}

\cwfocv*

\begin{proof}
By \Cref{lem:cw-weighted-fo-cost-value-optimization}, graph classes of cliquewidth at most $q$ admit efficient weighted \FO tuple optimization for every fixed $q$.
The claim follows from \Cref{lem:tuple-optimization-implies-focv}, with $q$ included in the parameter.
\end{proof}

Again, the local version is immediate. 

\begin{lemma}
\label{lem:local-cw-local-focv}
Every graph class of locally bounded cliquewidth admits efficient \textsc{Local} \FOCV.
\end{lemma}

\begin{proof}
Let $\Cc$ be a graph class of locally bounded cliquewidth.
For every fixed radius $s$, all induced neighborhoods $G[N_s^G[a]]$ with $G\in\Cc$ and $a\in V(G)$ have cliquewidth bounded by a constant depending only on $s$ and $\Cc$.
By \Cref{lem:cw-weighted-fo-cost-value-optimization}, these neighborhood classes admit efficient weighted \FO tuple optimization after adding a unary predicate for the center.
First-order model checking on $\Cc$ is fixed-parameter tractable by the standard local cliquewidth model checking theorem.
Therefore, \Cref{lem:local-neighborhood-optimization} applies.
\end{proof}

The following lemma follows from \Cref{lem:local-cw-local-focv}, as the multicolored distance-$r$ independent set is itself an \FO formula.

\begin{lemma}
\label{lem:local-cw-anchored-wmdi}
Every graph class of locally bounded cliquewidth admits efficient anchored weighted multicolored distance-$r$ independence for every fixed $r\in\N$.
\end{lemma}

\localcwfocv*

\begin{proof}
By \Cref{lem:local-cw-local-focv}, the class $\Cc$ admits efficient \textsc{Local} \FOCV.
By \Cref{lem:local-cw-anchored-wmdi}, it admits efficient anchored weighted multicolored distance-$r$ independence for every fixed $r \in \N$.
The result follows from \Cref{thm:conditional-FO-value-discovery}.
\end{proof}

\subsection{Monadically stable classes}
We finally prove the two monadic-stability results. The existential weighted result uses quasi-bounded-size bounded-shrubdepth decompositions. The full first-order result is for the Boolean optimization version and uses first-order model checking together with the anchored distance-independence lemma.

For the existential fragment, we first recall the decomposition machinery used for these graph classes. Let~$X$ be a graph invariant, such as treedepth, shrubdepth, cliquewidth, or rankwidth. 
The class $\Cc$ has \emph{quasi-bounded-size bounded-$X$ decompositions} if, for every $p\in\N$ and every $\varepsilon>0$, there is a number $d(p,\varepsilon)$ such that every graph $G\in\Cc$ admits a family $\mathcal U=\{U_1,\ldots,U_m\}$ of subsets of $V(G)$ satisfying the following conditions:
\begin{itemize}
    \item $m\le |V(G)|^\varepsilon$.
    \item For every set $A\subseteq V(G)$ with $|A|\le p$, there is an index $i\in[m]$ such that $A\subseteq U_i$.
    \item For every $i\in[m]$, the induced subgraph $G[U_i]$ has $X$-value at most $d(p,\varepsilon)$.
\end{itemize}
We call such decompositions \emph{efficiently computable} if, for every fixed choice of the parameters in the definition, the family $\mathcal U$ and certificates for the asserted bounds on the subgraphs $G[U_i]$ can be computed in time $f(p,\varepsilon)\cdot |V(G)|^{O(1)}$.
In the bounded-size case, the dependence on $\varepsilon$ is omitted.
Certificates may be treedepth decompositions, tree-models, cliquewidth expressions, or rank-decompositions, depending on the invariant $X$.

The sparsity theory of Ne\v{s}et\v{r}il and Ossona de Mendez shows that (monotone) nowhere dense classes are exactly the classes with quasi-bounded-size bounded-treedepth decompositions~\cite{GroheKS17,NesetrilOM12}.
For hereditary graph classes, Braunfeld et al.~characterize hereditary monadic stability by quasi-bounded-size bounded-shrubdepth decompositions~\cite{BraunfeldNOdMS25}.

The last characterization is presented as an existential statement in~\cite{BraunfeldNOdMS25}.
The proof combines sparse neighborhood covers with the existence of bounded-length winning strategies in the Flipper game.
The neighborhood covers were made algorithmic in the first-order model-checking work on monadically stable graph classes~\cite{DreierEMMPT23}, and the Flipper-game strategies can be computed efficiently by the algorithmic version of the Flipper-game characterization of monadically stable classes~\cite{GajarskyMMOPPST23}.
Substituting these algorithmic ingredients into the decomposition construction yields the efficiently computable quasi-bounded-size bounded-shrubdepth decompositions used here.

\begin{theorem}[Algorithmic version of \cite{GajarskyMMOPPST23}]
\label{thm:alg-low-shrubdepth-covers-monadic-stable}
Let $\Cc$ be a monadically stable graph class.
Then~$\Cc$ admits efficiently computable quasi-bounded-size bounded-shrubdepth decompositions.
\end{theorem}

All the mentioned graph-classes come with first-order model checking meta-theorems.
On nowhere dense classes, \FO model checking is fixed-parameter tractable by the theorem of Grohe, Kreutzer, and Siebertz~\cite{GroheKS17}.
On monadically stable graph classes, \FO model checking is fixed-parameter tractable by the work of Dreier and coauthors~\cite{DreierEMMPT23,DMS23}.

Tuple optimization on bounded shrubdepth is fixed-parameter tractable as a consequence of the fact that bounded shrubdepth classes have structurally bounded expansion and \Cref{lem:sbe-weighted-fo-cost-value-optimization}.  

\existentialfovaluemonadicstable*

\begin{proof}
Let an instance of the existential fragment of \FOCV on a graph $G\in\Cc$ be given.
Apply \Cref{lem:focv-cost-normalization} to the cost functions.
If the normalization rejects, then no feasible tuple exists, and we reject.
We keep the notation $c_i$ for the normalized costs and write~$B$ for the normalized budget.
We then consider one finite case of the value expression at a time.
Let this case be $(\varphi_j,\alpha_j)$, and write the existential formula as $\varphi_j(\bar x)\equiv\exists\bar y\,\theta_j(\bar x,\bar y)$ with~$\theta_j$ quantifier-free.
Let $\bar x=(x_1,\ldots,x_k)$, let $\ell\coloneqq|\bar y|$, and put $p\coloneqq k+\ell$.
Let $T_j\coloneqq W-\alpha_j$.
Let $\gamma_i^j$ and $\omega_i^j$ be the finite cost and value functions produced by \Cref{lem:finite-threshold-tuple-reduction} for the formula~$\varphi_j(\bar x)$, the normalized costs, the value functions, the budget~$B$, and the threshold $T_j$.

Choose a fixed $\varepsilon>0$, for instance $\varepsilon=1$.
Compute the quasi-bounded-size bounded shrubdepth decomposition from \Cref{thm:alg-low-shrubdepth-covers-monadic-stable} for the parameter $p$ and this value of $\varepsilon$.
For every cluster~$U$ of the decomposition, we optimize over tuples $(\bar x,\bar y)\in U^{k+\ell}$ inside the induced graph~$G[U]$.
For the selected coordinates $x_i$, we use the finite value functions $\omega_i^j$ and the finite cost functions $\gamma_i^j$.
For the witness coordinates in $\bar y$, we use value zero and cost zero.
We apply \Cref{lem:sbe-weighted-fo-cost-value-optimization} to the quantifier-free formula $\theta_j(\bar x,\bar y)$ on the bounded-shrubdepth graph $G[U]$ and budget $B$.
We keep the best modified coordinate value over all clusters for this case.

For soundness, suppose that the best value for the current case is at least $T_j$.
Then, for some cluster $U$, the optimizer returns tuples $(\bar a,\bar b)$ with $G[U]\models\theta_j(\bar a,\bar b)$, modified cost at most $B$, and modified coordinate value at least $T_j$.
Quantifier-free truth is absolute between $G[U]$ and $G$ for tuples contained in $U$, so $G\models\theta_j(\bar a,\bar b)$ and hence, $G\models\varphi_j(\bar a)$.
By the construction in \Cref{lem:finite-threshold-tuple-reduction}, the tuple $\bar a$ has a finite original value sum of at least $T_j$ and a normalized cost of at most $B$.
Therefore, its original \FOCV value is at least $T_j+\alpha_j=W$, and \Cref{lem:focv-cost-normalization} transfers the cost bound back to the original instance.

For completeness, let $\bar a$ be a feasible tuple for the considered case, and let $\bar b$ be witnesses such that $G\models\theta_j(\bar a,\bar b)$.
The set of vertices appearing in $(\bar a,\bar b)$ has size at most $k+\ell=p$, so it is contained in some cluster $U$ of the decomposition.
When this cluster is processed, the optimizer considers the tuple $(\bar a,\bar b)$.
By \Cref{lem:focv-cost-normalization}, the normalized cost of $\bar a$ is at most~$B$.
Since $\sum_i w_i(a_i)\ge T_j$ and this value is finite, the modified coordinate value coincides with the original coordinate value by \Cref{lem:finite-threshold-tuple-reduction}.
Thus, the best value for this case reaches~$T_j$.

The instance is accepted if and only if the best value reaches $T_j$ for at least one finite case.
The decomposition has at most $f(\Cc,p,\varepsilon)\cdot n^\varepsilon$ clusters, and every cluster is processed in polynomial time for fixed parameters.
Thus, the running time is fixed-parameter tractable.
\end{proof}

For the Boolean full first-order fragment, we need the following anchored distance-independence algorithm.
It follows from \Cref{cor:existential-FO-value-monadically-stable}, since the multicolored distance-$r$ independence condition is \FO definable.

\begin{lemma}
\label{lem:monadic-stable-anchored-wmdi}
Let $\Cc$ be a monadically stable graph class.
The anchored weighted multicolored distance-$r$ independence problem is fixed-parameter tractable on $\Cc$, parameterized by the number $p$ of colors and by the radius~$r$.
\end{lemma}

Finally, we solve the non-weighted case of \textsc{\FO Discovery} on monadically stable classes.
The key idea of the Boolean monadically stable case is that the numerical information needed locally can be encoded by boundedly many unary colors. 
After the movement costs have been capped, we fix a center and a bounded offset pattern, add unary predicates for the corresponding distance layers, and then test the resulting local formulas by first-order model checking on a monadic expansion of the input graph. 
Thus the weighted information is absorbed into the formulas, while the global choice of local candidates is handled by anchored distance independence.

\booleanfodiscoverymonadicstable*

\begin{proof}
Let $\psi(\bar x)$ be the Boolean target formula, let $\bar s=(s_1,\ldots,s_k)$ be the initial tuple, and let $b$ be the movement budget.
Put $B\coloneqq\min\left\{b,k(|V(G)|-1)\right\}$.
If a token and a target vertex lie in different connected components, their distance is infinite, and they cannot appear in a solution of finite budget.
Therefore, it is enough to cap all finite distances at $B+1$.
We enumerate all permutations $\pi\in S_k$ and solve the labeled problem in which coordinate $i$ is paid for by the source token $s_{\pi(i)}$.

Fix one permutation $\pi$.
Apply the refined free-variable Gaifman normal form to $\psi(\bar x)$.
Let $r$ be the locality radius.
We enumerate a graph $H\in\mathcal H_k$ and an index $\tau\in\Lambda_H$.
The tuple-independent side condition $\xi_H^\tau$ is tested by first-order model checking on the monadically stable class.
If it is false, we discard this choice.
Let $\mathcal P(H)=\{B_1,\ldots,B_p\}$ be the connected components of $H$.
For every component $B_h$, let $\rho_h\coloneqq\min B_h$.
Define $\widehat r\coloneqq\max\left\{1,r\right\}$, $r^\ast\coloneqq(k+1)\widehat r$, $\sigma_0\coloneqq0$, $D_\ell\coloneqq(k-1)\sigma_\ell+r^\ast$, $R_\ell\coloneqq2D_\ell+r$, and $\sigma_{\ell+1}\coloneqq R_\ell$ for $\ell<k$.
These are the same numbers as in the proof of \Cref{thm:conditional-FO-value-discovery}.

We enumerate a number $\ell\in\{0,\ldots,p-1\}$ and a forest $F$ on the component set $[p]$.
Let $\mathcal Q(F)$ be the set of connected components of $F$.
For every cluster $C\in\mathcal Q(F)$, put $I_C\coloneqq\bigcup_{h\in C}B_h$ and $\eta_C\coloneqq\min C$.
We use the same cluster formula $\chi_{H,\tau,\ell,F,C}(\bar x_{I_C},z_C)$ as in the proof of \Cref{thm:conditional-FO-value-discovery}.
It contains the radius-$r$ pattern induced by $H$ on $I_C$, the local formulas $\psi_{H,B_h}^\tau$ for $h\in C$, the forest edges witnessing scale-$\sigma_\ell$ connectivity, the equation $z_C=x_{\rho_{\eta_C}}$, and the protection conditions $\dist_G(z_C,x_i)\le D_\ell$ for $i\in I_C$.

We now construct candidates for the cluster $C$.
For every center $c\in V(G)$ and every offset vector $(\Delta_i)_{i\in I_C}\in\{-D_\ell,\ldots,D_\ell\}^{I_C}$, define $d_i\coloneqq\dist_G(s_{\pi(i)},c)+\Delta_i$.
If some source~$s_{\pi(i)}$ is disconnected from $c$, or if some $d_i$ is not in $\{0,\ldots,|V(G)|-1\}$, we skip this choice.
Otherwise, let $L_{i,d_i}\coloneqq\{v\in V(G):\dist_G(s_{\pi(i)},v)=d_i\}$.
Using first-order model checking on the expansion of $G$ by a singleton predicate $Z_c$ for $c$ and by the finitely many layer predicates~$L_{i,d_i}$, we test whether the formula
\begin{equation*}
        \exists z_C\exists \bar x_{I_C}\left(Z_c(z_C)\wedge\chi_{H,\tau,\ell,F,C}(\bar x_{I_C},z_C)\wedge\bigwedge_{i\in I_C}L_{i,d_i}(x_i)\right)
\end{equation*}
holds.
The number of added predicates depends only on $k$, and monadic stability is preserved under such unary expansions.
Whenever the test succeeds, we create a candidate for color $C$, anchored at $c$, with cost $\sum_{i\in I_C}d_i$ and profit zero.
The offset range is bounded by a function of $k$ and $\psi$, so only $f(k,\psi)\cdot |V(G)|$ candidates are produced for each cluster.

It remains to choose one candidate for each cluster so that the anchors are pairwise at a distance greater than $R_\ell$, and the total cost is at most $B$.
This is an anchored weighted multicolored distance-$R_\ell$ independence instance with all profits equal to zero and a profit threshold of zero.
By \Cref{lem:monadic-stable-anchored-wmdi}, it is fixed-parameter tractable.

We prove soundness.
Assume that the anchored instance has a feasible choice of candidates.
For every cluster $C$, the corresponding model-checking test gives a local realization $\bar b_C$ of the formula $\chi_{H,\tau,\ell,F,C}$ whose coordinate $i$ lies in the distance layer $L_{i,d_i}$.
Combining these local realizations gives a tuple $\bar b=(b_1,\ldots,b_k)$.
The layer predicates ensure that the movement cost of coordinate $i$ from $s_{\pi(i)}$ is exactly $d_i$ and hence, the total movement cost is the sum of the chosen candidate costs.
This sum is at most $B\le b$.
The protection conditions and the anchored distance-$R_\ell$ independence guarantee that variables belonging to distinct clusters are pairwise at distance greater than $r$.
Inside each cluster, the formula $\delta_{H[I_C],r}$ enforces the radius-$r$ pattern induced by $H$.
Therefore, $H=H_G^r(\bar b)$.
All local formulas $\psi_{H,B_h}^\tau$ hold by construction, and the side condition $\xi_H^\tau$ was tested to be true.
By the refined Gaifman normal form, $G\models\psi(\bar b)$.
Thus, the initial placement can be moved within budget $b$ to a feasible tuple.

We prove completeness.
Suppose that there is a tuple $\bar b=(b_1,\ldots,b_k)$ satisfying $G\models\psi(\bar b)$ and reachable from the initial placement within budget $b$.
Choose a permutation $\pi$ witnessing the minimum movement cost.
For this permutation, let $H\coloneqq H_G^r(\bar b)$ and let $\tau$ be the unique normal-form index certified by \Cref{thm:refined-free-variable-gaifman}.
As in the number observation in the proof of \Cref{thm:conditional-FO-value-discovery}, choose $\ell\in\{0,\ldots,p-1\}$ such that $Q_{\sigma_\ell}(\bar b)$ and $Q_{R_\ell}(\bar b)$ have the same connected components, and choose a forest $F$ whose components are these common components and whose edges lie in $Q_{\sigma_\ell}(\bar b)$.
For every cluster $C\in\mathcal Q(F)$, set $c_C\coloneqq b_{\rho_{\eta_C}}$.
Every vertex $b_i$ with \mbox{$i\in I_C$} lies within a distance of at most $D_\ell$ from $c_C$.
Thus, $\Delta_i\coloneqq\dist_G(s_{\pi(i)},b_i)-\dist_G(s_{\pi(i)},c_C)$ lies in $\{-D_\ell,\ldots,D_\ell\}$.
The model-checking test for the center $c_C$ and this offset vector succeeds, because it is witnessed by $\bar b_{I_C}$.
Hence, the corresponding candidate is created.
Candidates belonging to distinct clusters have anchors at distance greater than $R_\ell$.
Their costs sum to the movement cost of $\bar b$ under $\pi$, which is at most $B$.
Therefore, the anchored distance-independence instance is feasible.

We repeat the construction over all permutations, all normal-form alternatives, all scales, and all forests.
Their number depends only on $k$ and on $\psi$.
Every model-checking call uses a fixed formula over a monadic expansion of a monadically stable graph, and every anchored independence instance is fixed-parameter tractable by \Cref{lem:monadic-stable-anchored-wmdi}.
Thus, Boolean \textsc{FO Discovery} is fixed-parameter tractable on $\Cc$.
\end{proof}

\bibliography{main}
\end{document}